\documentclass{article}

\PassOptionsToPackage{table}{xcolor}
\usepackage{iclr2027_conference}
\iclrfinalcopy
\usepackage{times}
\usepackage{latexsym}
\usepackage[T1]{fontenc}
\usepackage[utf8]{inputenc}
\usepackage{microtype}
\usepackage{inconsolata}

\usepackage{amsmath,amsfonts,amssymb,amsthm,mathtools}
\usepackage{dsfont}

\usepackage{graphicx}
\newcommand{\fitmath}[1]{\adjustbox{max width=\columnwidth}{$\displaystyle #1$}}
\usepackage{booktabs}
\usepackage{multirow}
\usepackage{multicol}
\usepackage{subcaption}
\usepackage[linesnumbered,ruled,vlined]{algorithm2e}

\usepackage{enumitem}
\usepackage{float,placeins,adjustbox}

\usepackage{hyperref}
\hypersetup{colorlinks=false,linkbordercolor={1 0 0},citebordercolor={0 1 0},urlbordercolor={0 1 1},pdfborder={0 0 1}}
\usepackage{url}
\usepackage{xurl}

\newtheorem{theorem}{Theorem}[section]
\newtheorem{proposition}[theorem]{Proposition}
\newtheorem{corollary}[theorem]{Corollary}
\newtheorem{lemma}[theorem]{Lemma}
\theoremstyle{definition}

\newtheorem{assumption}[theorem]{Assumption}
\theoremstyle{remark}

\numberwithin{equation}{section}

\title{%
  When AI Reviews Its Own Code: Recursive Self-Training Collapse in Code LLMs
}

\author{\textbf{Xinyuan Song}$^{1}$ \quad
    \textbf{Zekun Cai}$^{2}$ \quad
    \textbf{Liang Zhao}$^{1}$ \\
    $^{1}$Emory University, Atlanta, GA, USA \quad
    $^{2}$The University of Tokyo, Tokyo, Japan \\
    \texttt{\{xinyuan.song,liang.zhao\}@emory.edu, caizekun@csis.u-tokyo.ac.jp} \\
}

\begin{document}
\raggedbottom

\maketitle

\begin{abstract}
Recursive self-training can degrade neural generative models when generated data is reused without fresh human data or external quality control. We study this risk in \textbf{code LLMs}, where AI-generated code can enter real repositories, later become training data, and create a repository-scale self-training loop. While software development traditionally interrupts this loop through \emph{pull-request review}, tests, compilation, and human approval, AI coding tools now produce code faster than humans can review it, and code review itself is increasingly automated by AI systems. We therefore compare three recursive fine-tuning regimes: \textbf{no review}, \textbf{human-style exogenous-gate review} using model-independent filters such as compilation and static quality checks, and \textbf{AI-self-gate review} using the code LLM's own signals such as perplexity and binary self-scoring. Across multiple code LLMs and benchmarks, every tested recursive regime falls below its own starting point; exogenous filters preserve executable validity and sometimes reduce decline, while \emph{AI-self-gate} filters can look strong early but later lose their filtering effect. In the clearest case, the binary self-gate enters a rubber-stamp regime where acceptance scores rise while benchmark correctness falls. We explain this behavior by formulating review as gated distributional reweighting, proving that AI self-gating degenerates to ungated self-training under a self-confirming acceptance condition, and giving a spectral analysis of representation-level covariance concentration under recursive retraining. These results suggest that stable recursive code LLM training requires exogenous verification rather than model-coupled self-review. 
\end{abstract}


\section{Introduction}

\label{sec:intro}
 
Recursive self-training has been repeatedly shown to be harmful for neural generative models when the loop is not anchored by fresh human data or an external quality signal. In image generation, self-consuming training can reduce distributional variance and drive models toward low-diversity outputs, a failure mode described as model amplification disorder~\cite{alemohammad2023selfconsuminggenerativemodelsmad}. In general generative modeling, training on recursively generated data can cause model collapse: distributional errors accumulate, low-probability modes are lost, and the learned distribution drifts away from the original data distribution~\cite{shumailov2024curserecursiontraininggenerated,dohmatob2024taleoftails}. Similar effects have been reported for language models, where synthetic text can reduce diversity, amplify earlier mistakes, and degrade long-horizon quality unless real data or strong external filtering is preserved in the training loop~\cite{briesch2023llmselfsuffering,seddik2024howbadsynthetic,suresh2024rateofcollapse}. These results give a general warning: self-training is not automatically self-improvement. When a model is trained on data produced by earlier versions of the same kind of model, the loop can become self-reinforcing rather than corrective.

Code large language models (LLMs)~\cite{benallal2023santacoder,li2023starcoder,roziere2023codellama,hui2024qwen25coder} are likely to face the same risk in a concrete and testable form. A code model can generate programs, add them back into the training corpus, and fine-tune the next model on this synthetic code. If these programs were always correct, diverse, and well reviewed, the loop could cheaply scale code data. In practice, generated code often contains bugs, incomplete logic, repeated templates, missing edge cases, and narrow stylistic patterns. Recursive training can therefore copy these errors into later models, making them better at imitating their own code style while worse at solving programming tasks. This is the recursive self-training trap for code: more data is created, but it can carry the model's own functional errors into future training.

This issue is no longer only an offline training concern. AI-generated code is rapidly entering real repositories through tools such as GitHub Copilot~\cite{github2026copilot}, Cursor~\cite{cursor2026}, Amazon CodeWhisperer~\cite{aws2026codewhisperer}, and Devin~\cite{cognition2026devin}. Studies report large productivity gains from Copilot~\cite{peng2023impactcopilot}, and enterprise reports show that Copilot-suggested code is frequently accepted, committed, and merged into pull requests~\cite{github2024copilotaccenture}. Future code LLMs may therefore be trained on repositories already containing AI-generated or AI-assisted code. If such code is collected without review, the loop becomes recursive self-training at repository scale: AI writes code, the code enters GitHub, GitHub enters the next training corpus, and the next model learns from earlier model outputs.

In real software development, code is usually filtered before entering a repository. A developer opens a pull request (PR), and the change may pass through compilation, tests, static analysis, and human review before merging. Modern code review is a lightweight but consequential quality-control process: it catches defects, transfers project knowledge, and affects downstream software quality~\cite{bacchelli2013expectations,bosu2015characteristics,mcintosh2016impact,sadowski2018modern}. This PR process acts as an external quality gate: the author does not define the acceptance rule, and the rule does not become weaker when the author writes worse code. Failed builds, failing tests, or poor design can still block the PR. We call this model-independent acceptance mechanism a \emph{human-style exogenous gate}. The experimental gates are automatic proxies, not literal manual review; they simulate aspects of the quality-feasible review target in Equation~\eqref{eq:quality_feasible_set}.

The AI coding era weakens this assumption. AI tools can generate PRs faster than humans can review them, while AI code review is becoming a practical engineering workflow. Recent studies examine LLM-assisted review with AI co-reviewers and interactive review assistants~\cite{mohajer2025rethinkingcodereview}; large-scale GitHub Actions analyses show AI review tools automatically commenting on real PRs~\cite{guo2025aicodereviewchanges}; and GitHub Copilot code review is now integrated into the PR workflow~\cite{github2026copilotcodereview}. At the same time, LLM-as-judge systems are known to be sensitive to evaluator design and can favor model-like outputs~\cite{zheng2023judging,panickssery2024selfpreference}. Future repositories may therefore contain both AI-generated and AI-reviewed code. In this setting, the verifier may no longer be a human-quality external process, but another AI model, possibly close to the generator. If AI writes the code and AI reviews the code, the acceptance signal can become coupled to the same distribution being retrained. The review process may look like filtering, but a biased AI reviewer can gradually accept the generator's degraded outputs as normal.

This observation leads to the central question of this paper: \textbf{In recursive code LLM training, how do no review, human-style exogenous review, and AI-based review differ?} Review changes the training distribution: without review, the model trains on all generated outputs; with review, training is reweighted toward samples that satisfy an acceptance rule. The key question is who defines this rule. A \emph{human-style exogenous gate} is fixed and model-independent, so it can reject degraded outputs as the generator worsens. An \emph{AI self-gate} is coupled to the generator, so it can drift with the model and accept degraded model style as high quality. This concern is consistent with evidence that LLM reviewers still miss many human-flagged PR issues, even when they are strong code generators~\cite{kumar2026sweprbench,tong2024codejudge,zhao2025codejudgeeval}.

We study this question through recursive fine-tuning experiments that directly compare no review, human-style exogenous-gate review, and AI-self-gate review. The exogenous-gate setting uses external, model-independent checks to decide which generated code enters future training. The AI-self-gate setting uses the code LLM itself to evaluate its own generated code, so the same model both produces and filters the training data.

We validate this comparison across multiple code LLM families and benchmarks. Our primary endpoints are HumanEval+ and MBPP+~\cite{liu2024humaneval}, with HumanEval~\cite{human_eval,openai_humaneval} and MBPP~\cite{mbpp} retained as secondary comparable benchmarks and LiveCodeBench~\cite{livecodebench} as a robustness probe. Our primary experiments use SantaCoder~\cite{benallal2023santacoder}, a 1.1B-parameter code model trained on permissively licensed code from The Stack~\cite{kocetkov2022stack}. We further test StarCoder2~\cite{lozhkov2024starcoder2}, Qwen2.5-Coder~\cite{hui2024qwen25coder}, and Code Llama~\cite{roziere2023codellama}, to examine whether the same pattern holds across model families.

\begin{figure}[!ht]
  
  \centering
  \includegraphics[width=0.75\linewidth]{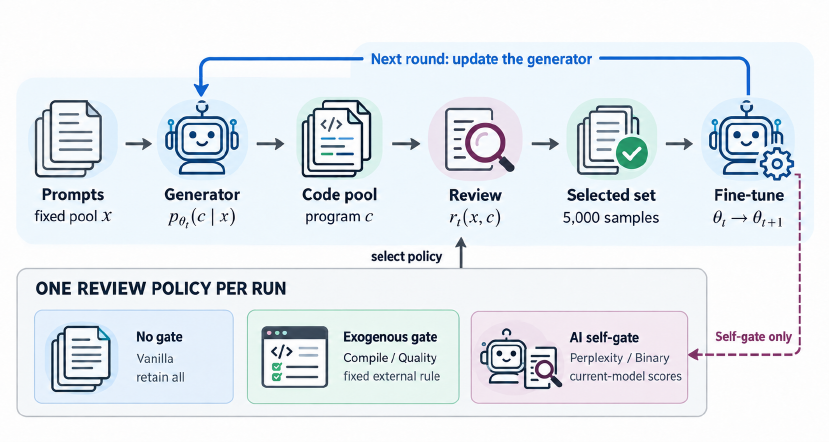}

  \caption{\textbf{Recursive training under different review policies.} Each round generates a code pool, applies one review policy, selects 5,000 samples, and fine-tunes the generator. Unlike fixed exogenous gates, AI self-gates are coupled to the current model.}
  \label{fig:training_loops}

\end{figure}

The results show a consistent below-baseline pattern across review regimes, with task-dependent reversals in which gate is best at a particular endpoint. In particular, unfiltered recursive fine-tuning degrades all tested model families; exogenous filters such as compile checks and quality rules preserve executable validity but cannot prevent long-horizon semantic drift; and AI-self-gate filters such as perplexity and binary self-scoring become less discriminative as the generator degrades. The binary classifier is the clearest AI-self-gate failure case: acceptance scores rise while benchmark correctness falls, indicating a self-confirming rubber-stamp regime.

Our contributions are summarized as follows:
\begin{enumerate}[leftmargin=1.5em,itemsep=1pt]
  \item \textbf{Code LLM self-training collapse:}
    We provide a systematic empirical study of recursive self-training collapse in code LLMs, covering five primary review regimes on SantaCoder, extended appendix trajectories for additional gate variants, and cross-model validation on StarCoder2, Qwen2.5-Coder, and Code Llama.

  \item \textbf{Review as gated reweighting:}
    We formulate review-based recursive training as gated sampling, where generated samples are reweighted by an acceptance function before entering the next round.
    This explains why useful gates can slow collapse relative to unfiltered self-training.

  \item \textbf{Exogenous gate vs.\ AI self-gate:}
    We identify exogeneity as the key difference between review regimes: human-style exogenous gates use fixed, model-independent criteria, while AI self-gates use scores coupled to the generator and can become self-confirming.

  \item \textbf{Theory of collapse and gate failure:}
    We prove that AI self-gate training degenerates to ungated self-training under a self acceptance condition (Theorem~\ref{thm:gating_degenerate}).
    We also provide a spectral analysis showing representation covariance concentration under recursive retraining (Proposition~\ref{prop:spectral}).
\end{enumerate}


\section{Problem Formulation and Theoretical Analysis}

\label{sec:framework}
\label{subsec:main_theory_preview}

\subsection{Iterative Self-Training Without vs.\ With Gating}

\label{subsec:ungated_gated}
\label{subsec:formal_gated_recursion}

Figure~\ref{fig:training_loops} illustrates the recursive training pipeline considered in this paper. Let $x\in\mathcal X$ denote a prompt or programming task, and let $c\in\mathcal C$ denote a candidate code solution. At each round, the generator produces a code pool that is processed by one of three review policies: no gate, an exogenous compile/quality gate, or an AI self-gate coupled to the current model. The selected samples are then reused to fine-tune the generator for the next round.

\textbf{No verification (ungated recursion).}
In ungated recursive self-training, every generated sample is directly reused as training data. The synthetic data distribution at iteration $t$ is
\begin{equation}
m_t^{\mathrm{ungated}}(x,c):=p_X(x)p_{\theta_t}(c\mid x).
\label{eq:ungated_dist}
\end{equation}
The next model is obtained by maximum-likelihood estimation (\textbf{MLE}) training on this distribution:
\begin{equation}
\theta_{t+1}^{\mathrm{ungated}}\in \arg\max_{\theta}\mathbb{E}_{(x,c)\sim m_t^{\mathrm{ungated}}}\big[\log p_{\theta}(c\mid x)\big].
\label{eq:ungated_mle}
\end{equation}
Equivalently, this minimizes $\mathrm{KL}(m_t^{\mathrm{ungated}}\|p_Xp_\theta)$. Thus, ungated recursion only enforces self-consistency: the next model is trained to imitate the current model's samples. It does not impose any external constraint on correctness, compilability, or execution success.

\textbf{Verification (gated recursion).}
In gated recursive self-training, generated code is accepted only after passing a review or verification signal. Let $\mathcal{E}(x,c)\in\{0,1\}$ denote the acceptance event, and define
\begin{equation}\label{eq:acceptance_score}
r(x,c):=\mathbb{P}(\mathcal{E}(x,c)=1\mid x,c)\in[0,1].
\end{equation}
Here, $r(x,c)$ can represent a compile check, execution test, static quality filter, human review score, or AI-based review score. Conditioning the generator on acceptance gives the accepted distribution
\begin{equation}
q_{\theta_t}(c\mid x):=\mathbb{P}(c\mid x,\mathcal{E}=1)=\frac{p_{\theta_t}(c\mid x)r(x,c)}{Z_{\theta_t}(x)},\quad Z_{\theta_t}(x):=\sum_{c'\in\mathcal C}p_{\theta_t}(c'\mid x)r(x,c').
\label{eq:gated_conditional}
\end{equation}
The gated synthetic data distribution is therefore
\begin{equation}
m_t^{\mathrm{gated}}(x,c):=p_X(x)q_{\theta_t}(c\mid x).
\label{eq:gated_dist}
\end{equation}
Training on accepted samples yields
\begin{equation}
\theta_{t+1}^{\mathrm{gated}}\in \arg\max_{\theta}\mathbb{E}_{(x,c)\sim m_t^{\mathrm{gated}}}\big[\log p_{\theta}(c\mid x)\big].
\label{eq:gated_mle}
\end{equation}
Thus, gated recursion differs from ungated recursion only through the sampling measure: before training the next model, generated samples are reweighted by the acceptance signal $r(x,c)$.

\textbf{Difference} (see Figures~\ref{fig:training_loops}).
The two recursions differ only in the sampling measure:
\begin{equation}
m_t^{\mathrm{gated}}(x,c) \propto m_t^{\mathrm{ungated}}(x,c)\cdot r(x,c).
\label{eq:gated_reweighting}
\end{equation}
Thus, gated recursion trains on the same generated distribution as ungated recursion, but reweights samples by acceptance: high-$r(x,c)$ code receives more training mass, while low-$r(x,c)$ code is suppressed.
At a fixed point of the gated recursion (pure self-training), one must satisfy
\begin{equation}
p_{\theta^*}(c\mid x)=q_{\theta^*}(c\mid x) \Longleftrightarrow p_{\theta^*}(c\mid x)\propto p_{\theta^*}(c\mid x)r(x,c),
\label{eq:gated_fixed_point}
\end{equation}
which implies that (for each $x$) probability mass concentrates on the support where $r(x,c)$ is maximal (e.g., $r(x,c)=1$ if such code exists). In contrast, the ungated recursion enforces only self-reproduction of its own samples and does not bias toward executability.
This is why a useful exogenous gate can slow collapse, while a gate that becomes constant on the generator's support loses its filtering effect and approaches ungated self-training.




\subsection{Representation-Subspace Drift Under Recursive Self-Training}

\label{subsec:subspace_drift}
 
We next describe a representation-level view of recursive self-training collapse. Let $\mathcal{D}_t$ denote the training distribution at iteration $t$ over prompt--code pairs $(x,c)\in\mathcal{X}\times\mathcal{C}$, where $\mathcal{D}_t$ may be ungated or gated. Fix a representation map
\begin{equation}
\phi:\mathcal{X}\times\mathcal{C}\to\mathbb{R}^d,\qquad z:=\phi(x,c).
\label{eq:rep_map}
\end{equation}
The distribution $\mathcal{D}_t$ induces a mean vector and covariance matrix in representation space:
\begin{equation}
\mu_t:=\mathbb{E}_{(x,c)\sim\mathcal{D}_t}[z],\qquad \Sigma_t:=\mathbb{E}_{(x,c)\sim\mathcal{D}_t}\big[(z-\mu_t)(z-\mu_t)^\top\big]\in\mathbb{R}^{d\times d}.
\label{eq:rep_mean_cov}
\end{equation}
The covariance $\Sigma_t$ captures which representation directions remain diverse under the current training distribution. If recursive self-training narrows the generated code distribution, then $\Sigma_t$ should concentrate on fewer dominant directions.

Let $U_t\in\mathbb{R}^{d\times k}$ contain the top-$k$ eigenvectors of $\Sigma_t$, and let $P_t:=U_tU_t^\top$ be the corresponding rank-$k$ projector. We measure consecutive subspace drift by
\begingroup\small
\begin{equation}
\mathrm{Sim}(t,t+1):=\frac{1}{k}\mathrm{Tr}(P_tP_{t+1})\in[0,1],\quad \mathrm{Dist}(t,t+1):=\|P_t-P_{t+1}\|_2=\sin\theta_{\max}(U_t,U_{t+1}).
\label{eq:subspace_sim_dist}
\end{equation}
\endgroup
Here, high similarity means the dominant representation subspace is stable, while large distance indicates that recursive training has shifted the main directions of variation.

To analyze the mechanism, we use a linearized covariance recursion: 
\begin{equation}
\Sigma_{t+1}=A\Sigma_tA^\top,\qquad \Sigma_0\succeq 0.
\label{eq:noiseless}
\end{equation}
This recursion should be read as a local approximation of how one round of retraining transforms representation covariance. Let $A=U\mathrm{diag}(s_1,\ldots,s_d)V^\top$ be an SVD with $s_1\ge s_2\ge\cdots\ge s_d\ge 0$. Then $\Sigma_t=A^t\Sigma_0(A^\top)^t$ and therefore
\begin{equation}
\|\Sigma_t\|_2\le \|A^t\|_2^2\|\Sigma_0\|_2=s_1^{2t}\|\Sigma_0\|_2.
\label{eq:spec_norm_growth}
\end{equation}
More importantly, different singular directions are amplified at different rates. For any unit vector $v$,
\begin{equation}
v^\top\Sigma_t v=\|(A^\top)^tv\|_{\Sigma_0}^2,\qquad \|(A^\top)^tv\|\ \text{scales as}\ s_i^t\ \text{when }v=v_i.
\label{eq:directional_variance_growth}
\end{equation}
Thus, if $s_1\gg s_2$, the component aligned with the leading singular direction dominates after repeated retraining. Whenever $\Sigma_0$ has nonzero energy in this direction, the ratio $\lambda_1(\Sigma_t)/\lambda_2(\Sigma_t)$ is expected to increase with $t$ because the leading component is amplified roughly by $s_1^{2t}$ while the next component is amplified roughly by $s_2^{2t}$. This gives a representation-space explanation for collapse: recursive self-training can progressively concentrate variance into a low-dimensional subspace, reducing diversity in the generated code distribution.

We now formalize the concentration effect by the linearized covariance recursion in Equation~\eqref{eq:noiseless}.

\begin{proposition}[Spectral concentration]
\label{prop:spectral}
Under the covariance recursion in Equation~\eqref{eq:noiseless}, recall that $\Sigma_t=A^t\Sigma_0(A^\top)^t$. Let $A=U\mathrm{diag}(s_1,\ldots,s_d)V^\top$ be the SVD of $A$, with $s_1>s_2\ge\cdots\ge s_d\ge 0$, and define $\widetilde{\Sigma}_0:=V^\top\Sigma_0V\succeq 0$. Then $\Sigma_t=U\mathrm{diag}(s_1^t,\ldots,s_d^t)\widetilde{\Sigma}_0\mathrm{diag}(s_1^t,\ldots,s_d^t)U^\top$. If $(\widetilde{\Sigma}_0)_{11}=v_1^\top\Sigma_0v_1\ge\alpha>0$ and $\rho_2:=s_2/s_1<1$, then
\begin{equation}
\frac{\lambda_1(\Sigma_t)}{\lambda_2(\Sigma_t)}\ge\frac{\alpha}{\lambda_2(\widetilde{\Sigma}_0)+\|\Sigma_0\|_2\rho_2^{2t}}.
\label{eq:spectral_ratio}
\end{equation}
Consequently, if $\lambda_2(\widetilde{\Sigma}_0)>0$, then $\lambda_1(\Sigma_t)/\lambda_2(\Sigma_t)\to\alpha/\lambda_2(\widetilde{\Sigma}_0)$ as $t\to\infty$; if $\lambda_2(\widetilde{\Sigma}_0)=0$, the ratio grows without bound.
\end{proposition}

The proof is given in Appendix~\ref{app:proofs}. Proposition~\ref{prop:spectral} shows that recursive retraining can amplify the dominant representation direction faster than all others, at a relative rate controlled by $\rho_2^{2t}$. This gives a representation-level mechanism for mode concentration: the model distribution can become increasingly low-rank even while training continues to optimize self-generated samples.

\subsection{Human Gate vs.\ AI Self-Gate}

\label{subsec:gate_taxonomy}
 
\subsubsection{Verifier parameterization: exogenous Human gate vs.\ endogenous AI self-gate}
 
We now specialize the generic acceptance function $r(x,c)$ from Equation~\eqref{eq:acceptance_score}. The key is whether the acceptance rule is fixed outside the generator or coupled to the generator during recursive training.

\paragraph{Human gate: exogenous acceptance.}
A Human gate uses a fixed acceptance function $r_H(x,c)\in[0,1]$, $r_H$ is independent of $t$ and $\theta_t$. It induces the accepted conditional distribution
\begin{equation}
q^H_{\theta_t}(c\mid x):=\frac{p_{\theta_t}(c\mid x)r_H(x,c)}{Z^H_{\theta_t}(x)},\quad Z^H_{\theta_t}(x):=\sum_{c'\in\mathcal C}p_{\theta_t}(c'\mid x)r_H(x,c').
\label{eq:human_gated_conditional}
\end{equation}
The corresponding training distribution is $m_t^H(x,c)=p_X(x)q^H_{\theta_t}(c\mid x)$. This covers model-independent review signals such as compilation, execution, quality checks, or human PR review.

\paragraph{AI self-gate: endogenous acceptance.}
An AI self-gate uses a learned reviewer with parameters $\phi_t$, producing $r_{\phi_t}(x,c)\in[0,1]$. It induces
\begin{equation}
q^A_{\theta_t,\phi_t}(c\mid x):=\frac{p_{\theta_t}(c\mid x)r_{\phi_t}(x,c)}{Z^A_{\theta_t,\phi_t}(x)},\quad Z^A_{\theta_t,\phi_t}(x):=\sum_{c'\in\mathcal C}p_{\theta_t}(c'\mid x)r_{\phi_t}(x,c').
\label{eq:ai_gated_conditional}
\end{equation}
The corresponding training distribution is $m_t^A(x,c)=p_X(x)q^A_{\theta_t,\phi_t}(c\mid x)$. In our experiments, this case corresponds to perplexity filtering and binary self-scoring, where the code LLM evaluates its own generated samples. Unlike $r_H$, the score $r_{\phi_t}$ can drift as the generator changes.

\paragraph{Taxonomy summary.}
Table~\ref{tab:gate_taxonomy} classifies the six filtering strategies by whether the acceptance rule is exogenous or endogenous. Vanilla has no gate and accepts all generated code. Compile is a Human gate because it directly runs Python compilation and accepts code only when compilation succeeds. Quality is also a Human gate because it uses fixed code-quality rules, described in Section~\ref{app:quality_feasible}, such as repetition rate and length. Compile+Quality combines these two model-independent rules, so it is also exogenous. Perplexity and Binary Classifier are AI self-gates: both use the code LLM itself to score its own generated code. Perplexity filters by the model's own likelihood, while Binary Classifier uses the model's own logit-difference score, defined in Section~\ref{app:selfeval}. Since the generator also provides the filtering signal, these two gates are coupled to $\theta_t$.

\begin{table}[!ht]

\centering
\caption{Human gate vs.\ AI self-gate taxonomy of all filtering strategies studied.}

\label{tab:gate_taxonomy}
\small
\begin{tabular}{llll}
\toprule
\hline
\textbf{Method} & \textbf{Gate type} & \textbf{Acceptance signal} & \textbf{Coupled to $\theta_t$?} \\
\midrule
Vanilla            & -- (no gate)     & Accept all                         & No  \\
Compile            & Human gate       & Compilability ($\texttt{compile}()$) & No  \\
Quality            & Human gate       & Repetition rate + length (no compile) & No \\
Compile+Quality    & Human gate       & Compile + repetition/length filters   & No \\
Perplexity         & AI self-gate     & Model perplexity (lowest 25\%)     & Yes \\
Binary Classifier  & AI self-gate     & Logit-difference score (top 25\%)  & Yes \\
\hline
\bottomrule
\end{tabular}

\end{table}

\paragraph{Coupled dynamics under AI self-gating.}
We now formalize the failure mode where the verifier is coupled to the generator. This captures the AI self-gate setting: the code LLM generates candidate code, scores its own outputs, and then trains on the samples it accepted. For a clean statement, we analyze the extreme case where the generator and verifier share the same parameters:
\begin{equation}
\phi_t=\theta_t,\qquad r_{\phi_t}(x,c)=r_{\theta_t}(x,c).
\label{eq:self_gate_shared_params}
\end{equation}
\begin{assumption}[Self-confirming acceptance]
\label{assump:self_confirming}
There exists a measurable function $\kappa:\mathcal{X}\to(0,1]$ such that for $p_X$-almost every $x$, $r_{\theta_t}(x,c)=\kappa(x)$ for $p_{\theta_t}(\cdot\mid x)$-almost every $c$.
\end{assumption}
Assumption~\ref{assump:self_confirming} describes the rubber-stamp regime: on the code that the current model is likely to generate, the AI reviewer assigns essentially the same acceptance score. Once this happens, the gate no longer changes the training distribution.
\begin{theorem}[Degeneracy to ungated recursion]
\label{thm:gating_degenerate}
Assume the shared-parameter self-gate in Equation~\eqref{eq:self_gate_shared_params} and Assumption~\ref{assump:self_confirming}. Then for $p_X$-almost every $x$,
\begin{equation}
Z^A_{\theta_t,\theta_t}(x):=\sum_{c'\in\mathcal{C}}p_{\theta_t}(c'\mid x)r_{\theta_t}(x,c')=\kappa(x),
\label{eq:degenerate_Z}
\end{equation}
and for $p_{\theta_t}(\cdot\mid x)$-almost every $c$,
\begin{equation}
q^A_{\theta_t,\theta_t}(c\mid x):=\frac{p_{\theta_t}(c\mid x)r_{\theta_t}(x,c)}{Z^A_{\theta_t,\theta_t}(x)}=p_{\theta_t}(c\mid x).
\label{eq:degenerate_q}
\end{equation}
Consequently,
\begin{equation}
m_t^A(x,c):=p_X(x)q^A_{\theta_t,\theta_t}(c\mid x)=p_X(x)p_{\theta_t}(c\mid x)=m_t^{\mathrm{ungated}}(x,c)\qquad\text{a.e.},
\label{eq:degenerate_m}
\end{equation}
and the gated \textbf{MLE} update coincides with the ungated self-training update:
\begin{equation}
\arg\max_{\theta}\mathbb{E}_{(x,c)\sim m_t^A}[\log p_\theta(c\mid x)]=\arg\max_{\theta}\mathbb{E}_{(x,c)\sim m_t^{\mathrm{ungated}}}[\log p_\theta(c\mid x)].
\label{eq:degenerate_update}
\end{equation}
\end{theorem}

The proof is given in Appendix~\ref{app:proofs}. Theorem~\ref{thm:gating_degenerate} states that once an AI self-gate becomes constant on the generator's own support, filtering becomes mathematically identical to no filtering. This is exactly the recursive self-training failure case where the model writes code, approves its own code, and then trains on that approved code.

\begin{corollary}[Collapsed fixed point]
\label{cor:collapsed_fp}
If there exists $(\theta^\star,\phi^\star)$ such that $r_{\phi^\star}(x,c)\equiv 1$ for all $(x,c)$, then for all $x$,
\begin{equation}
q^A_{\theta^\star,\phi^\star}(c\mid x)=p_{\theta^\star}(c\mid x),\qquad m^A(x,c)=p_X(x)p_{\theta^\star}(c\mid x).
\label{eq:collapsed_fixed_point}
\end{equation}
Thus, the gating operator is the identity on $p_{\theta^\star}(\cdot\mid x)$.
\end{corollary}

The proof is given in Appendix~\ref{app:proofs}.

\subsection{When AI review matches human review.}

The self-confirming case above characterizes when AI review degenerates to ungated recursion. Conversely, AI-gated training remains close to exogenous human-gated training when the AI reviewer stays uniformly calibrated to the human reviewer and the human gate retains non-vanishing acceptance mass.

\begin{lemma}[Approximation of human-gated review]
\label{lem:equiv}
Define $\varepsilon_t:=\sup_{x\in\mathcal X,c\in\mathcal C}\big|r_{\phi_t}(x,c)-r_H(x,c)\big|$. Suppose $Z^H_{\theta_t}(x)\ge z_0>0$ for $p_X$-almost every $x$. If $\varepsilon_t<z_0$, then
\begin{equation}
\big\|q^A_{\theta_t,\phi_t}(\cdot\mid x)-q^H_{\theta_t}(\cdot\mid x)\big\|_1\le\frac{2\varepsilon_t}{z_0},
\label{eq:main_gate_approx}
\end{equation}
and consequently $\|m_t^A-m_t^H\|_1\le 2\varepsilon_t/z_0$.
\end{lemma}

Lemma~\ref{lem:equiv} shows that calibration to an external review signal controls the deviation of AI-gated recursion from human-gated recursion. The full formulation, sufficient conditions for $z_0>0$, and proof are given in Appendix~\ref{app:formal_theory}.

\subsection{Quality Gate}

\label{subsec:quality_feasible}

The ideal gate would accept exactly the programs in the \emph{quality gate set} $\mathcal{F}_{\mathrm{qual}}(q)$, defined by quality dimensions (A)--(G) in Appendix~\ref{app:quality_feasible}. Our Human-gate methods approximate this ideal using fixed, model-independent checks: compilation for basic runtime validity and quality rules for repetition and length. In contrast, AI self-gates rely on model-derived scores that can drift with the generator and lose alignment with $\mathcal{F}_{\mathrm{qual}}(q)$, as characterized in Theorem~\ref{thm:gating_degenerate}.


\section{Experiments}

\label{sec:experiments}


\subsection{Evaluation Setup}

\label{subsec:setup}

\paragraph{Evaluation setup.}
We use \textbf{HumanEval+} and \textbf{MBPP+}~\cite{liu2024humaneval} as the primary endpoints, with \textbf{HumanEval}~\cite{human_eval,openai_humaneval} and \textbf{MBPP}~\cite{mbpp} as secondary benchmarks and \textbf{LiveCodeBench}~\cite{livecodebench} as a robustness probe. We evaluate \textbf{SantaCoder} (1.1B)~\cite{benallal2023santacoder}, \textbf{StarCoder2-3B}~\cite{lozhkov2024starcoder2}, \textbf{Qwen2.5-Coder-1.5B}~\cite{hui2024qwen25coder}, and \textbf{Code Llama-7B}~\cite{roziere2023codellama} for five recursive rounds. All evaluations use \texttt{bigcode-evaluation-harness}~\cite{bigcode_eval_harness} with temperature 0.8, top-$p$ 0.95, and $n=200$ candidates per problem for pass@1.

\paragraph{Recursive training.}
At each round, the current model generates code from a fixed prompt pool, the review policy selects $5{,}000$ samples, and the next model is fine-tuned for 3{,}000 steps. Exogenous gates generate until $5{,}000$ samples are accepted, whereas AI self-gates generate 20k candidates and retain 25\%. SantaCoder uses pure-code fine-tuning, while StarCoder2 and Qwen2.5-Coder use approximately 70--80\% code and 20--30\% code-related natural language. All methods otherwise share the same training configuration. Because the sampling procedures differ, we separately report raw gate pass rate, accepted-set execution validity, and candidates required per 5k accepted. Full details are in Appendix~\ref{app:impl}.

\paragraph{Filtering strategies.}
We evaluate five primary regimes (Table~\ref{tab:gate_taxonomy}). \textbf{No gate:} (1)~\textbf{Vanilla} accepts every generated sample. \textbf{Exogenous gates:} (2)~\textbf{Compile} accepts code that passes Python \texttt{compile()}; (3)~\textbf{Quality} applies fixed model-independent criteria, requiring repetition rate $\leq0.3$ and length $\geq50$ tokens. \textbf{AI self-gates:} (4)~\textbf{Perplexity} generates 20k candidates and retains the lowest 25\% by masked-completion perplexity; (5)~\textbf{Binary Classifier} retains the top 25\% by the model's logit-difference score between \texttt{good} and \texttt{bad}. Thus, the regimes differ in whether the training data are unfiltered, filtered by a fixed external rule, or filtered by a score coupled to the recursively trained model. Extended Compile+Quality and Quality-20R variants are reported in Appendix~\ref{app:per_round_tables}.

\begin{table}[!ht]

\centering
\caption{Cross-model Round~0 and Round~5 performance. For each model, we report Vanilla recursion and the strongest exogenous (EG) and AI self-gate according to R5 HumanEval+. \textbf{Bold} indicates the strongest reported R5 value within each model and metric.}

\label{tab:cross_model_summary}
\scriptsize
\setlength{\tabcolsep}{2pt}
\resizebox{0.8\linewidth}{!}{%
\begin{tabular}{llccccc}
\toprule
\textbf{Model} & \textbf{Stage} & \textbf{HumanEval} & \textbf{HumanEval+}
  & \textbf{MBPP} & \textbf{MBPP+} & \textbf{LCB} \\
\midrule
SantaCoder (1.1B)  & Baseline & 0.189 & 0.171 & 0.352 & 0.294 & 0.018 \\
StarCoder2-3B      & Baseline & 0.317 & 0.274 & 0.595 & 0.492 & 0.000 \\
Qwen2.5-Coder-1.5B & Baseline & 0.451 & 0.372 & 0.698 & 0.582 & 0.238 \\
Code Llama-7B      & Baseline & 0.328 & 0.287 & 0.441 & 0.388 & 0.045 \\
\midrule
SantaCoder (1.1B)  & R5 Vanilla      & 0.079 & 0.073 & 0.027 & 0.019 & 0.000 \\
SantaCoder (1.1B)  & R5 Compile (EG) & 0.092 & 0.079 & 0.045 & 0.034 & 0.000 \\
SantaCoder (1.1B)  & R5 Binary (AI)  & \textbf{0.116} & \textbf{0.104} & \textbf{0.079} & \textbf{0.058} & 0.000 \\
\midrule
StarCoder2-3B      & R5 Vanilla      & 0.128 & 0.104 & \textbf{0.185} & \textbf{0.172} & 0.000 \\
StarCoder2-3B      & R5 Quality (EG) & 0.104 & 0.092 & 0.106 & 0.098 & \textbf{0.003} \\
StarCoder2-3B      & R5 Binary (AI)  & \textbf{0.140} & \textbf{0.122} & 0.090 & 0.066 & 0.000 \\
\midrule
Qwen2.5-Coder-1.5B & R5 Vanilla      & 0.049 & 0.043 & 0.085 & 0.082 & 0.000 \\
Qwen2.5-Coder-1.5B & R5 Compile (EG) & 0.104 & 0.098 & \textbf{0.148} & \textbf{0.127} & \textbf{0.003} \\
Qwen2.5-Coder-1.5B & R5 Binary (AI)  & \textbf{0.146} & \textbf{0.122} & 0.082 & 0.071 & 0.000 \\
\midrule
Code Llama-7B      & R5 Vanilla      & 0.177 & 0.134 & \textbf{0.235} & \textbf{0.214} & \textbf{0.003} \\
Code Llama-7B      & R5 Compile (EG) & \textbf{0.189} & \textbf{0.171} & 0.124 & 0.090 & \textbf{0.003} \\
Code Llama-7B      & R5 PPL (AI)     & 0.165 & 0.134 & 0.206 & 0.175 & 0.000 \\
\bottomrule
\end{tabular}%
}

\end{table}

\subsection{Main Collapse Results}

\label{subsec:main_results}

\paragraph{Cross-model collapse patterns.}
Figure~\ref{fig:collapse_overview} summarizes the recursive trajectories across all four model families and the two gated review categories. Despite substantial differences in initial performance, every model moves below its own Round~0 baseline over five recursive rounds on both HumanEval+ and MBPP+. The figure also shows that gating changes the shape of degradation, but neither exogenous nor AI self-gating consistently preserves the initial performance level.

\begin{figure}[!ht]

  \centering
  \includegraphics[width=0.8\linewidth]{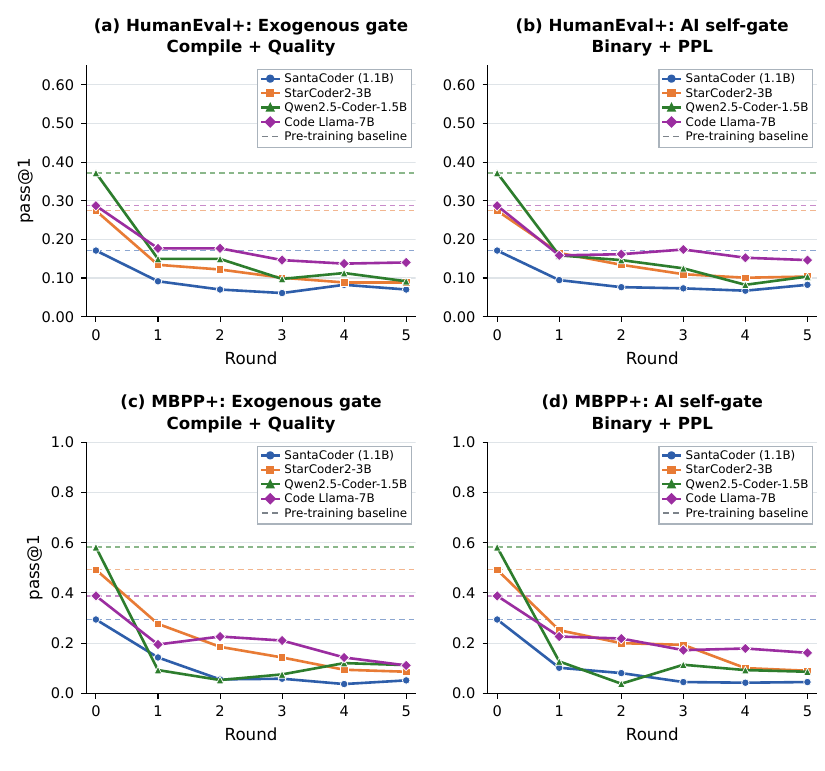}

  \caption{%
    \textbf{Cross-model recursive degradation under different review policies.}
    Trajectories from Round~0 to Round~5 for all four models.
    Left: exogenous gates (Compile + Quality average); right: AI self-gates (Binary + PPL average).
    Top: HumanEval+; bottom: MBPP+.
    Dashed horizontal lines denote model-specific Round~0 baselines.
  }
  \label{fig:collapse_overview}
  
\end{figure}

\paragraph{Different models collapse at different rates.}
The degradation rate varies across models and benchmarks. Qwen2.5-Coder and SantaCoder show particularly sharp declines in some settings, while Code Llama is relatively more stable. This variation suggests that recursive robustness is not determined by initial model capability alone and depends on the model--benchmark combination.

\paragraph{Review policies change the endpoint.}
Table~\ref{tab:cross_model_summary} shows that no review policy consistently performs best across models and benchmarks. Different gates improve Round~5 retention in different settings, but all reported strategies remain below their corresponding Round~0 baselines. This indicates that gating primarily changes the severity of degradation rather than preventing recursive collapse.

\subsection{Effect of Review Gates}
\label{subsec:mechanistic_checks}

\paragraph{Execution-backed review.}
Table~\ref{tab:execution_backed_review} examines the samples that enter recursive training. On Qwen2.5-Coder-1.5B at Round~5, Compile raises accepted-set execution validity from 33.4\% to 89.8\%, whereas the AI self-gate retains 32.0\%. Thus, exogenous execution filtering substantially changes the accepted training distribution, but does not eliminate downstream degradation.

\begin{table}[!ht]

\centering
\caption{Execution-backed review on Qwen2.5-Coder-1.5B at Round~5.}

\label{tab:execution_backed_review}
\resizebox{0.8\textwidth}{!}{%
\begin{tabular}{lcccccc}
\toprule
\textbf{Gate} & \shortstack{\textbf{Raw gate}\\\textbf{pass}} & \shortstack{\textbf{Accepted}\\\textbf{execution}} & \shortstack{\textbf{Candidates}\\\textbf{/ 5k}} & \shortstack{\textbf{HumanEval}\\\textbf{$S_5/S_0$}} & \shortstack{\textbf{MBPP}\\\textbf{$S_5/S_0$}} & \textbf{$\Delta_{0\to5}$} \\
\midrule
No review & 1.00 & 33.4\% & 5{,}000 & 0.108 & 0.121 & $-0.150$ \\
Compile only & 0.898 & 89.8\% & 5{,}568 & 0.230 & 0.212 & $-0.109$ \\
AI self-gate & 0.25 & 32.0\% & 20{,}000 & 0.324 & 0.117 & $-0.126$ \\
\bottomrule
\end{tabular}%
}

\end{table}

\subsection{Filter Comparison at Round 1 and Round 5}
\label{app:filter_comparison}

Figure~\ref{fig:filter_comparison} compares HumanEval+ performance across filters at Round~1 and Round~5.

\begin{figure}[!ht]

  \centering
  \includegraphics[width=0.86\textwidth]{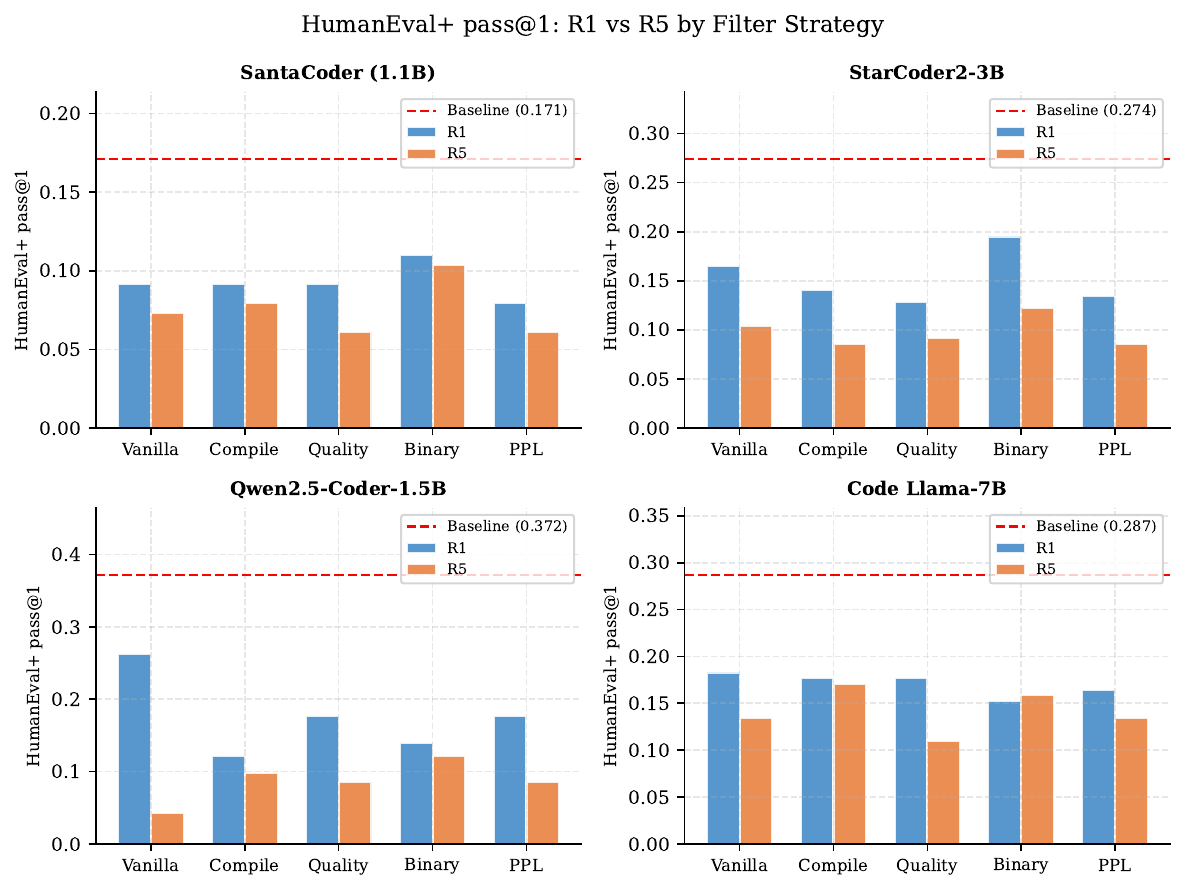}
  
  \caption{\textbf{HumanEval+ pass@1 at Round~1 and Round~5 across filtering strategies.} Blue bars denote R1, orange bars denote R5, and red dashed lines denote pre-training baselines.}
  \label{fig:filter_comparison}
  
\end{figure}

Binary achieves the highest R5 HumanEval+ for SantaCoder, StarCoder2-3B, and Qwen2.5-Coder, while Compile is strongest for Code Llama-7B. However, all R5 results remain below their pre-training baselines, indicating that filter choice changes the degradation trajectory but does not eliminate recursive collapse.


\section{Conclusion}

\label{sec:conclusion}
 
We studied recursive self-training collapse in code LLMs under no review, exogenous-gate review, and AI-self-gate review. Across models and benchmarks, all tested regimes fall below their starting baselines. Exogenous gates preserve external validity signals but do not prevent degradation, while AI self-gates remain coupled to the evolving generator and can become self-confirming. Gating nevertheless changes both the accepted training distribution and the resulting degradation trajectory. Our theory characterizes when model-coupled review degenerates toward ungated recursion. Our results suggest that stable recursive training requires verification signals that remain externally grounded.

\subsection*{AI use statement}

Generative AI tools were used to assist with LaTeX formatting, template conversion, and the organization of manuscript material between the main text and appendix.

\bibliographystyle{iclr2027_conference}
\bibliography{reference}

\clearpage

\appendix
\makeatletter\setlength{\@fptop}{0pt}\makeatother


\section{Related Work}

\label{app:related}

\subsection{Model Collapse}

The systematic study of model collapse begins with Shumailov et al.~\cite{shumailov2024curserecursiontraininggenerated}, who demonstrate that training on recursively generated data causes (1) early collapse: distributional errors accumulate and the model drifts from the true distribution; and (2) late collapse: low-frequency events permanently disappear.
Their experiments on VAE, GMM, and OPT confirm that fresh human data must be periodically injected to prevent collapse.
Alemohammad et al.~\cite{alemohammad2023selfconsuminggenerativemodelsmad} study self-consuming image generation models, proving that pure synthetic training loops monotonically shrink distribution variance (MAD: Model Amplification Disorder).
They identify three loop types: fully synthetic, augmented (partial human mixing), and accumulating (fixed replay); only augmented and accumulating loops with sufficient human fraction avoid collapse.

Briesch et al.~\cite{briesch2023llmselfsuffering} provide the first LLM-specific self-consuming study using logically verifiable outputs (logical expressions), finding that synthetic data fraction is the primary collapse predictor.
Seddik et al.~\cite{seddik2024howbadsynthetic} provide a statistical analysis showing that pure synthetic recursion leads to collapse, while sufficiently mixing real data can stabilize the process.
Dohmatob et al.~\cite{dohmatob2024taleoftails} reframe collapse as a change in scaling laws---synthetic training truncates the power-law tail of the data distribution, eliminating rare events.
Suresh et al.~\cite{suresh2024rateofcollapse} prove that preserving a fraction $\epsilon > 0$ of real data per generation keeps the KL divergence from the true distribution bounded, while $\epsilon \to 0$ (pure synthetic) leads to unbounded divergence---formalizing the role of exogenous quality signal.

On the constructive side, Ferbach et al.~\cite{ferbach2024curateddatapreferences} show that curation according to a reward model can turn iterative retraining into an implicit preference-optimization mechanism.
This directly justifies our exogenous-gate framework: a compile gate or execution gate aligns with ground-truth correctness; AI self-gate does not.
Gillman et al.~\cite{gillman2024selfcorrecting} propose a correction function that maps synthetic samples toward higher-probability regions of the true distribution, stabilizing loops; code execution is the natural correction function in the code domain.
Zhu et al.~\cite{zhu2024howtosynthesize} demonstrate token-level diversity perturbations as a mitigation strategy.
Fu et al.~\cite{fu2025preventcollapse} introduce recursive stability conditions for Transformer in-context learning.

\subsection{Modern Code Review and Exogenous Verification}

Modern code review provides the software-engineering analogue of the exogenous gate in our framework.
Bacchelli and Bird~\cite{bacchelli2013expectations} show that defect finding is only one part of review: reviewers also enforce design norms, share knowledge, and improve maintainability.
Bosu et al.~\cite{bosu2015characteristics} study which review comments developers find useful, while McIntosh et al.~\cite{mcintosh2016impact} link code-review coverage, participation, and reviewer expertise to post-release software quality.
At industrial scale, Sadowski et al.~\cite{sadowski2018modern} describe Google's tool-based review process over millions of changes.
These studies motivate our abstraction of review as an exogenous acceptance rule: the important property is not that every decision is manual, but that acceptance is not produced by the same model distribution being retrained.

\subsection{Code LLM Development}

\textbf{Representation era (2020--2021).}
CodeBERT~\cite{feng2020codebert} established the bimodal NL-PL representation paradigm (encoder-only, 6 languages, MLM + RTD).
Codex~\cite{human_eval} introduced generative code LLMs (12B, fine-tuned GPT on GitHub) and the pass@$k$ evaluation standard.
AlphaCode~\cite{li2022alphacode} scaled to competition-level difficulty at 41B parameters, using massive sampling ($10^3$ candidates) plus execution-based filtering---an early exogenous gate.

\textbf{Fill-in-the-middle and open-source (2022--2023).}
InCoder~\cite{fried2022incoder} unified generation and infilling via causal masking.
CodeGen~\cite{nijkamp2022codegen} demonstrated multi-turn multi-lingual generation at 16B.
The Stack~\cite{kocetkov2022stack} established 3TB permissive-license code as the community standard.
SantaCoder~\cite{benallal2023santacoder} showed that 1.1B models with high-quality filtering can match much larger counterparts; this is our primary base model.
StarCoder~\cite{li2023starcoder} (15.5B) set the open-source SOTA on HumanEval.
CodeT5+~\cite{wang2023codet5plus} demonstrated flexible encoder-decoder pretraining objectives.

\textbf{Data-quality and instruction-tuning era (2023--2024).}
WizardCoder~\cite{luo2023wizardcoder} applied Evol-Instruct (AI-generated instruction evolution) to code, reaching HumanEval 57.3\%; this is an early AI-assisted data augmentation pipeline.
Gunasekar et al.~\cite{gunasekar2023textbooks} showed that 6B tokens of GPT-3.5 synthetic data can train a 1.3B model to HumanEval 50.6\% (phi-1)---one-shot synthetic, not recursive.
Magicoder~\cite{wei2023magicoder} introduced OSS-Instruct: seed real code snippets $\to$ GPT-generated problems $\to$ 7B model at HumanEval 70.7\%.
StarCoder2~\cite{lozhkov2024starcoder2} and The Stack v2 extended to 900+ languages with improved governance.
CodeRL~\cite{le2022coderl} uses unit-test feedback and a learned critic for program synthesis, showing the value of grounding code generation in functional signals.
SelfCodeAlign~\cite{wei2024selfcodealign} introduced the first execution-validated self-alignment pipeline; its execution gate corresponds to our exogenous-gate archetype.
Code Llama~\cite{roziere2023codellama} (7B, 13B, 34B) fine-tuned Llama 2 on code with a mixed code-language pretraining data mixture.

\textbf{State-of-the-art (2024--2025).}
DeepSeek-Coder~\cite{guo2024deepseekcoder} (1B--33B) achieves HumanEval 79.3\% with repository-level training.
DeepSeek-Coder-V2~\cite{zhu2024deepseekcoderv2} (MoE 236B/21B active) matches GPT-4-Turbo code ability.
Qwen2.5-Coder~\cite{hui2024qwen25coder} (0.5B--32B) reports open-source SOTA code performance, with the 32B-Instruct model matching GPT-4o-level coding ability.
OpenCoder~\cite{huang2024opencoder} provides fully open recipe.
OpenCodeInterpreter~\cite{zheng2024opencodeinterpreter} integrates code generation with execution and refinement, further emphasizing executable feedback as an external training signal.
RLEF~\cite{gehring2024rlef} grounds code LLMs in execution feedback with RL, providing complementary evidence for execution-backed training signals.
DeepSeek-R1~\cite{deepseek2025r1} demonstrates GRPO-based reasoning scaling with rule-based verifiable rewards, reaching top-96th-percentile on Codeforces.


\section{Evaluation Benchmarks: Full Description}

\label{app:benchmarks}

\paragraph{HumanEval.}
Each problem consists of a function signature, docstring, and 7--9 unit tests.

\paragraph{MBPP.}
The Mostly Basic Programming Problems dataset~\cite{mbpp} contains 374 crowd-sourced Python
problems focusing on introductory-level tasks (string manipulation, arithmetic, sorting).

\paragraph{LiveCodeBench.}
LiveCodeBench~\cite{livecodebench} continuously collects new programming problems from LeetCode,
AtCoder, and Codeforces, preventing contamination from models that have been trained on earlier
problems.
In our experiments, SantaCoder achieves near-zero pass@1 on LiveCodeBench throughout retraining,
reflecting the gap between introductory-level proficiency and competition-level reasoning.

\paragraph{Code generation metrics.}
Table~\ref{tab:codegen_metrics_app} summarizes common metrics for evaluating code generation quality.
\begin{table}[!ht]
\centering
\renewcommand{\arraystretch}{1.25}
\resizebox{\textwidth}{!}{%
\begin{tabular}{p{3.6cm} p{2.8cm} p{9.2cm}}
\toprule
\hline
\textbf{Metric} & \textbf{Measured on} & \textbf{Typical definition} \\
\midrule

\textbf{pass@k (functional correctness)}~\cite{human_eval} &
$k$ sampled candidates per problem &
A problem is counted as solved if \emph{at least one} of the $k$ sampled programs passes the full unit-test suite; reported as pass@1, pass@10, pass@100.~\cite{human_eval} \\
\hline
\textbf{pass@1 (single-sample correctness)}~\cite{human_eval} &
$n$ sampled candidates per problem &
Special case of the unbiased pass@k estimator with $k{=}1$; with $n=200$ samples, this is the mean single-sample correctness over the sampled candidates.~\cite{human_eval} \\
\hline
\textbf{CodeBLEU}~\cite{ren2020codebleu} &
Generated code vs.\ reference code &
A weighted combination of (i) BLEU $n$-gram match, (ii) weighted $n$-gram match, (iii) AST match, and (iv) data-flow match.~\cite{ren2020codebleu} \\
\hline
\textbf{Text overlap}~\cite{papineni2002bleu,lin2004rouge,popovic2015chrf} &
Generated text/code vs.\ references &
Text overlap metrics: BLEU (precision-weighted $n$-gram overlap)~\cite{papineni2002bleu}, ROUGE (recall-oriented overlap family)~\cite{lin2004rouge}, and chrF (character $n$-gram F-score)~\cite{popovic2015chrf}. \\
\hline
\textbf{Diversity of samples}~\cite{paul2024benchmarks} &
Multiple candidates for the same prompt/problem &
Statistics over a set of sampled solutions, such as pairwise similarity distribution, distinct-$n$, or de-duplication rate, often reported together with pass@k to measure ``diverse and correct'' generation.~\cite{paul2024benchmarks} \\
\hline
\bottomrule
\end{tabular}%
}

  \caption{\textbf{Common metrics for model-generated code.}
The main experiments use pass@1 because recursive retraining changes single-sample correctness, while the table lists complementary metrics that capture functional correctness, lexical overlap, structural similarity, and sample diversity.}
\label{tab:codegen_metrics_app}

\end{table}
%
%

\paragraph{SantaCoder: Open code LMs with governance and infilling-capable training.}
The BigCode technical report on SantaCoder describes the development of 1.1B-parameter code LMs trained on permissively licensed data from The Stack with an emphasis on transparency and responsible release~\cite{benallal2023santacoder,kocetkov2022stack}. The report studies deployment-relevant modeling choices, notably Multi-Query Attention (MQA) for faster decoder inference~\cite{shazeer2019fast} and Fill-in-the-Middle (FIM) objectives that enable arbitrary-position code infilling~\cite{bavarian2022efficient,fried2022incoder}. Evaluation is conducted using MultiPL-E, which extends unit-test-driven code generation benchmarks to many languages~\cite{cassano2022multipl-e}.

\paragraph{StarCoder: Open-access code LMs with long context, infilling, and responsible release.}
StarCoder and StarCoderBase are 15.5B-parameter open-access code LMs trained primarily on permissively licensed data from The Stack~\cite{li2023starcoder,kocetkov2022stack}. The models combine an 8K context window, Fill-in-the-Middle (FIM) training for infilling, and Multi-Query Attention (MQA) for faster large-batch decoding~\cite{li2023starcoder,bavarian2022efficient,shazeer2019fast}.

\paragraph{StarCoder2: scaling open code LMs with The Stack v2.}
StarCoder2 provides 3B/7B/15B parameter checkpoints trained on 3.3--4.3T tokens, incorporating Fill-in-the-Middle (FIM) training and fast decoder inference via multi-query style attention~\cite{lozhkov2024starcoder2,bavarian2022efficient,shazeer2019fast}.


\section{Additional Robustness and Diagnostics}
\label{app:additional_diagnostics}

This section provides additional diagnostics of the recursive training results, including external anchors, corrupted exogenous gates, seed robustness, and early-versus-late filter behavior. These analyses test how the observed degradation changes under alternative review and training conditions.

\subsection{Effect of External Verified Anchors}
\label{app:external_anchor_ablation}

We first test whether adding a small amount of fresh verified data can stabilize recursive training. Table~\ref{tab:external_anchor_ablation} compares synthetic-only recursion, execution-filtered recursion, and execution-filtered recursion with a 5\% external anchor.

\begin{table}[!ht]
\centering
\caption{External anchor ablation at Round~5.}
\label{tab:external_anchor_ablation}
\resizebox{\textwidth}{!}{%
\begin{tabular}{lcccccc}
\toprule
\textbf{Selection policy} & \textbf{Anchor fraction} & \textbf{Raw gate pass} & \textbf{Candidates / 5k} & \textbf{$S_5/S_0$} & \textbf{Retention area} & \textbf{Accepted execution} \\
\midrule
Synthetic only, no gate & $0\%$ & 1.00 & 5{,}000 & 0.364 & 2.46 & 34.6\% \\
Execution-filtered synthetic & $0\%$ & 0.894 & 5{,}593 & 0.247 & 2.11 & 89.4\% \\
Execution-filtered + fresh anchor & $5\%$ & 0.898 & 5{,}290 + 250 anchor & 0.180 & 1.354 & 89.8\% \\
\bottomrule
\end{tabular}%
}
\end{table}

The execution-filtered conditions substantially increase accepted-set execution validity, from 34.6\% to approximately 90\%. However, the 5\% fresh anchor does not improve long-horizon retention in this experiment: its $S_5/S_0$ remains 0.180 and below the initial baseline. Thus, a small external anchor alone is insufficient to prevent recursive degradation.

\subsection{Sensitivity to Exogenous-Gate Noise}
\label{app:noisy_exogenous_gate}

We next test whether the behavior of an exogenous gate changes when its decisions are corrupted. Table~\ref{tab:noisy_exogenous_gate} injects Bernoulli decision flips at increasing noise levels.

\begin{table}[!ht]
\centering
\caption{Noisy exogenous gate summary.}
\label{tab:noisy_exogenous_gate}
\resizebox{0.88\textwidth}{!}{%
\begin{tabular}{lccccc}
\toprule
\textbf{Gate noise} & \textbf{Effective flip rate $\eta$} & \textbf{Acceptance mass $Z_t$} & \textbf{HumanEval $S_5/S_0$} & \textbf{MBPP $S_5/S_0$} & \textbf{$\Delta_{0\to5}$} \\
\midrule
$\sigma=0.00$ & 0.00 & 1.00 & 0.230 & 0.212 & $-0.109$ \\
flip $\eta=0.05$ & 0.050 & 0.535 & 0.041 & 0.106 & $-0.141$ \\
flip $\eta=0.10$ & 0.099 & 0.508 & 0.176 & 0.117 & $-0.126$ \\
flip $\eta=0.20$ & 0.198 & 0.495 & 0.257 & 0.049 & $-0.123$ \\
\bottomrule
\end{tabular}%
}
\end{table}

Injecting gate noise sharply changes both the effective acceptance mass and the resulting benchmark trajectory. The endpoints are not monotonic in the flip rate---for example, HumanEval retention recovers at higher noise while MBPP retention does not. We therefore interpret this experiment as evidence that recursive behavior is sensitive to gate calibration, rather than as a monotonic noise--performance relationship.

\subsection{Robustness Across Additional Runs}
\label{app:single_seed_robustness}

We additionally examine the available robustness runs for SantaCoder and Qwen2.5-Coder under the three core review regimes. Table~\ref{tab:single_seed_robustness} reports their Round~5 performance, retention area, and degradation slope.

\begin{table}[!ht]
\centering
\caption{Focused single-seed robustness summary at Round~5.}
\label{tab:single_seed_robustness}
\resizebox{0.92\textwidth}{!}{%
\begin{tabular}{llcccc}
\toprule
\textbf{Model} & \textbf{Regime} & \textbf{HumanEval R5} & \textbf{MBPP R5} & \textbf{Retention area} & \textbf{Slope $\Delta_{0\to5}$} \\
\midrule
SantaCoder-1.1B & No review & 0.079 & 0.027 & 2.166 & $-0.139$ \\
SantaCoder-1.1B & Exogenous gate & 0.092 & 0.045 & 2.106 & $-0.110$ \\
SantaCoder-1.1B & AI self-gate & 0.116 & 0.079 & 2.119 & $-0.107$ \\
Qwen2.5-Coder-1.5B & No review & 0.049 & 0.085 & 1.852 & $-0.150$ \\
Qwen2.5-Coder-1.5B & Exogenous gate & 0.104 & 0.148 & 1.380 & $-0.109$ \\
Qwen2.5-Coder-1.5B & AI self-gate & 0.146 & 0.082 & 1.648 & $-0.126$ \\
\bottomrule
\end{tabular}%
}
\end{table}

All reported runs still exhibit substantial degradation by Round~5. The relative gate behavior is model-dependent: AI self-gating gives the highest HumanEval endpoint for both models, whereas the exogenous gate gives the highest MBPP endpoint for Qwen2.5-Coder. This again supports the main finding that gating changes the degradation trajectory without producing a universally dominant review policy.



\section{Formal Gated-Training Analysis}

\label{app:formal_theory}

\paragraph{Human-style target and automatic proxies.}
The ideal human-style target is broader than any single automatic check. It would accept code that is semantically correct for the task, executes safely, avoids obvious inefficiency or brittle behavior, satisfies security and compliance constraints, and remains maintainable. The Quality-Feasible Code Set in Appendix~\ref{app:quality_feasible} formalizes this target as a conjunction of such requirements. Our implemented Compile and Quality gates are not literal manual review. They are automatic, model-independent proxies for parts of that target: compilation checks syntactic/executable validity, while the static quality filter removes short or repetitive generations. This distinction matters for interpretation. The experiments test whether a fixed external proxy signal behaves differently from a model-coupled self-review signal; they do not claim that the simple proxy fully matches a human pull-request reviewer.

\subsection{When AI review matches exogenous review: equivalence and a non-collapse condition}

\paragraph{Uniform calibration error.}
Recall the exogenous-gate conditional $q^H_{\theta_t}$ in Equation~\eqref{eq:human_gated_conditional} and the AI-self-gate conditional $q^A_{\theta_t,\phi_t}$ in Equation~\eqref{eq:ai_gated_conditional}. Here, $r_H(x,c)$ denotes the human-style external review rule, while $r_{\phi_t}(x,c)$ denotes AI review: the learned reviewer scores whether generated code should be accepted for future training. We measure how well the AI reviewer matches the exogenous reviewer by the worst-case calibration error
\begin{equation}
\varepsilon_t:=
\sup_{x\in\mathcal{X},c\in\mathcal{C}}
\left|r_{\phi_t}(x,c)-r_H(x,c)\right|.
\label{eq:eps_def}
\end{equation}

\paragraph{Non-collapse exogenous acceptance mass.}
To avoid a degenerate case where almost no generated code is accepted by the exogenous reviewer, we require the human-style gate to accept a non-vanishing fraction of the generator's probability mass.
\begin{assumption}[Non-collapse exogenous acceptance mass]
\label{assump:z0}
There exists $z_0\in(0,1]$ such that
\begin{equation}
\begin{aligned}
Z^H_{\theta_t}(x)
&=\sum_{c'}p_{\theta_t}(c'\mid x)r_H(x,c')\\
&\ge z_0
\quad\text{for $p_X$-a.e.\ $x$ and all $t$.}
\end{aligned}
\label{eq:z0}
\end{equation}
\end{assumption}
Assumption~\ref{assump:z0} says that the external reviewer still finds enough acceptable code in the generator distribution. Without this lower bound, the accepted distribution can become unstable because it is normalized by a vanishing acceptance mass.

\paragraph{Quality gate sufficient condition.}
A simple sufficient condition for Assumption~\ref{assump:z0} follows from the quality gate set $\mathcal{F}_{\mathrm{qual}}(x)$. Suppose the exogenous reviewer gives nontrivial acceptance probability to quality gate code: $c\in\mathcal{F}_{\mathrm{qual}}(x)\Rightarrow r_H(x,c)\ge\delta_H$ for some $\delta_H\in(0,1]$. If the generator also assigns nontrivial mass to quality gate code, $p_{\theta_t}(\mathcal{F}_{\mathrm{qual}}(x)\mid x)\ge\beta_H$ for some $\beta_H>0$, then Assumption~\ref{assump:z0} holds with $z_0=\delta_H\beta_H$, since
\begin{equation}
\begin{aligned}
Z^H_{\theta_t}(x)
&=\sum_{c'}p_{\theta_t}(c'\mid x)r_H(x,c')\\
&\ge
\sum_{c'\in\mathcal{F}_{\mathrm{qual}}(x)}
p_{\theta_t}(c'\mid x)r_H(x,c')\\
&\ge
\delta_H p_{\theta_t}\!\left(\mathcal{F}_{\mathrm{qual}}(x)\mid x\right)\\
&\ge\delta_H\beta_H.
\end{aligned}
\label{eq:z0_from_quality_feasible}
\end{equation}
Thus, the exogenous gate remains stable when it accepts quality gate code and the generator has not lost all probability mass on such code.

\begin{theorem}[Uniform approximation of exogenous-gated review]
\label{thm:equiv}
Suppose the AI reviewer satisfies Equation~\eqref{eq:eps_def} with $\varepsilon_t<z_0$, and suppose Assumption~\ref{assump:z0} holds. Then for $p_X$-almost every $x$,
\begin{equation}
\left\|
q^A_{\theta_t,\phi_t}(\cdot\mid x)
-q^H_{\theta_t}(\cdot\mid x)
\right\|_1
\le \frac{2\varepsilon_t}{z_0-\varepsilon_t}.
\label{eq:q_bound}
\end{equation}
Consequently,
\begin{equation}
\fitmath{\begin{aligned}
\left\|m_t^A-m_t^H\right\|_1
&=
\int_{\mathcal{X}}p_X(x)
\left\|
q^A_{\theta_t,\phi_t}(\cdot\mid x)
-q^H_{\theta_t}(\cdot\mid x)
\right\|_1dx\\
&\le \frac{2\varepsilon_t}{z_0-\varepsilon_t}.
\end{aligned}}
\label{eq:m_bound}
\end{equation}
\end{theorem}

The proof is in Appendix~\ref{app:proofs}. Theorem~\ref{thm:equiv} shows when AI review remains safe: the AI reviewer must stay close to exogenous review, and the exogenous gate must accept nonzero generator mass. If $\varepsilon_t$ grows or $z_0$ shrinks, the bound weakens and AI review can approach ungated self-training.

Equivalently, AI-gated recursion avoids collapse only when $\varepsilon_t\ll z_0$ and $r_H(x,\cdot)$ is not approximately constant on $\mathrm{supp}(p_{\theta_t}(\cdot\mid x))$. These conditions mean that AI review remains calibrated to an exogenous review standard, and that standard still separates good from bad generated code.

\section{Representation-Subspace Drift Under Recursive Self-Training}

\label{app:subspace_drift_details}
 
This recursion~\eqref{eq:noiseless} should be read as a local approximation of how one round of retraining transforms representation covariance. For the closed-form concentration statement below, we analyze the self-adjoint local case $A=U\mathrm{diag}(a_1,\ldots,a_d)U^\top$ with $a_1\ge a_2\ge\cdots\ge a_d\ge0$. Then $\Sigma_t=A^t\Sigma_0(A^\top)^t$ and therefore
\begin{equation}
\|\Sigma_t\|_2\le \|A^t\|_2^2\|\Sigma_0\|_2=a_1^{2t}\|\Sigma_0\|_2.
\label{eq:spec_norm_growth_self_adjoint}
\end{equation}
More importantly, different eigen directions are amplified at different rates. For any unit vector $v$,
\begin{equation}
\begin{aligned}
v^\top\Sigma_t v&=\|(A^\top)^tv\|_{\Sigma_0}^2,\\
\|(A^\top)^tv\|&\ \text{scales as}\ a_i^t\ \text{when }v=u_i.
\end{aligned}
\label{eq:directional_variance_growth_self_adjoint}
\end{equation}
Thus, if $a_1\gg a_2$, the component aligned with the leading eigen direction dominates after repeated retraining. Whenever $\Sigma_0$ has nonzero energy in this direction, the ratio $\lambda_1(\Sigma_t)/\lambda_2(\Sigma_t)$ is expected to increase with $t$ because the leading component is amplified roughly by $a_1^{2t}$ while the next component is amplified roughly by $a_2^{2t}$. This gives a representation-space explanation for collapse: recursive self-training can progressively concentrate variance into a low-dimensional subspace, reducing diversity in the generated code distribution.

We now formalize the concentration effect by the linearized covariance recursion in Equation~\eqref{eq:noiseless}.

\begin{proposition}[Spectral concentration]
\label{prop:spectral_self_adjoint}
Under the covariance recursion in Equation~\eqref{eq:noiseless}, assume the local operator is symmetric positive semidefinite, $A=U\mathrm{diag}(a_1,\ldots,a_d)U^\top$, with $a_1>a_2\ge\cdots\ge a_d\ge0$. Assume also that $\Sigma_0$ is diagonal in the same eigenbasis, $\widetilde{\Sigma}_0:=U^\top\Sigma_0U=\mathrm{diag}(\sigma_1^2,\ldots,\sigma_d^2)\succeq0$. If $\sigma_1^2\ge\alpha>0$ and $\rho_2:=a_2/a_1<1$, then whenever $\lambda_2(\Sigma_t)>0$,
\begin{equation}
\frac{\lambda_1(\Sigma_t)}{\lambda_2(\Sigma_t)}
\ge
\frac{\alpha}{\|\Sigma_0\|_2\rho_2^{2t}}.
\label{eq:spectral_ratio_self_adjoint}
\end{equation}
If $\lambda_2(\Sigma_t)=0$, the ratio is infinite. Thus the leading-to-second eigenvalue ratio grows at least as fast as $\rho_2^{-2t}$ up to the initial covariance scale.
\end{proposition}

The proof is given in Appendix~\ref{app:proofs}. Proposition~\ref{prop:spectral_self_adjoint} shows, under this self-adjoint local approximation, that recursive retraining can amplify the dominant representation direction faster than all others. This gives a representation-level mechanism for mode concentration: the model distribution can become increasingly low-rank even while training continues to optimize self-generated samples.


\section{Complete Proofs}

\label{app:proofs}

\subsection{Proof of Theorem~\ref{thm:gating_degenerate} (AI Self-Gate Degeneracy)}

\begin{proof}
Fix any $x$ such that Assumption~\ref{assump:self_confirming} holds and $\kappa(x)>0$ (this is $p_X$-a.e.\ $x$).
By Assumption~\ref{assump:self_confirming}, $r_{\theta_t}(x,c)=\kappa(x)$ for $p_{\theta_t}(\cdot\mid x)$-a.e.\ $c$.
Therefore,
\begin{equation}
\begin{aligned}
Z^A_{\theta_t,\theta_t}(x)
&=\sum_{c'\in\mathcal{C}}p_{\theta_t}(c'\mid x) r_{\theta_t}(x,c')\\
&=\sum_{c'\in\mathcal{C}}p_{\theta_t}(c'\mid x) \kappa(x)\\
&=\kappa(x)\sum_{c'\in\mathcal{C}}p_{\theta_t}(c'\mid x)\\
&=\kappa(x).
\end{aligned}
\end{equation}
For any $c$ in the $p_{\theta_t}(\cdot\mid x)$-support where $r_{\theta_t}(x,c)=\kappa(x)$,
\begin{equation}
\begin{aligned}
q^A_{\theta_t,\theta_t}(c\mid x)
&=\frac{p_{\theta_t}(c\mid x) r_{\theta_t}(x,c)}{Z^A_{\theta_t,\theta_t}(x)}\\
&=\frac{p_{\theta_t}(c\mid x) \kappa(x)}{\kappa(x)}\\
&=p_{\theta_t}(c\mid x),
\end{aligned}
\end{equation}
which proves $q^A_{\theta_t,\theta_t}(\cdot\mid x)=p_{\theta_t}(\cdot\mid x)$ as measures (for $p_X$-a.e.\ $x$).
Multiplying by $p_X(x)$ yields $m_t^{A}(x,c)=p_X(x)p_{\theta_t}(c\mid x)=m_t^{\mathrm{ungated}}(x,c)$ almost everywhere.
Substituting $m_t^{A}=m_t^{\mathrm{ungated}}$ into the respective MLE objectives shows the two argmax sets coincide.
\end{proof}

\subsection{Proof of Theorem~\ref{thm:equiv} (Uniform-Approximation Equivalence)}

\begin{proof}
Fix $x$ with $Z^H_{\theta_t}(x)\ge z_0$. First note that
\begin{equation}
\fitmath{\begin{aligned}
\big|Z^A_{\theta_t,\phi_t}(x)-Z^H_{\theta_t}(x)\big|
&=
\left|\sum_{c'}p_{\theta_t}(c'\mid x)\big(r_{\phi_t}(x,c')-r_H(x,c')\big)\right|\\
&\le \varepsilon_t,
\end{aligned}}
\end{equation}
hence $Z^A_{\theta_t,\phi_t}(x)\ge Z^H_{\theta_t}(x)-\varepsilon_t\ge z_0-\varepsilon_t>0$.
Now write
\begin{equation}
\fitmath{\begin{aligned}
q^A_{\theta_t,\phi_t}(c\mid x)-q^H_{\theta_t}(c\mid x)
&=
p_{\theta_t}(c\mid x)\left(
\frac{r_{\phi_t}(x,c)}{Z^A_{\theta_t,\phi_t}(x)}
-
\frac{r_H(x,c)}{Z^H_{\theta_t}(x)}
\right),
\end{aligned}}
\end{equation}
so by adding and subtracting $r_H(x,c)/Z^A_{\theta_t,\phi_t}(x)$ and using triangle inequality,
\begin{equation}
\fitmath{\begin{aligned}
\big\|q^A_{\theta_t,\phi_t}(\cdot\mid x)-q^H_{\theta_t}(\cdot\mid x)\big\|_1
&\le
\sum_c p_{\theta_t}(c\mid x)\frac{|r_{\phi_t}(x,c)-r_H(x,c)|}{Z^A_{\theta_t,\phi_t}(x)}\\
&+
\sum_c p_{\theta_t}(c\mid x) r_H(x,c)\left|\frac{1}{Z^A_{\theta_t,\phi_t}(x)}-\frac{1}{Z^H_{\theta_t}(x)}\right|\\
&\le
\frac{\varepsilon_t}{Z^A_{\theta_t,\phi_t}(x)}
+
Z^H_{\theta_t}(x)\cdot
\frac{|Z^A_{\theta_t,\phi_t}(x)-Z^H_{\theta_t}(x)|}{Z^A_{\theta_t,\phi_t}(x) Z^H_{\theta_t}(x)}\\
&\le
\frac{\varepsilon_t}{z_0-\varepsilon_t}
+
\frac{\varepsilon_t}{z_0-\varepsilon_t}\\
&=
\frac{2\varepsilon_t}{z_0-\varepsilon_t},
\end{aligned}}
\end{equation}
where $\varepsilon_t<z_0$ ensures the denominator is positive.
Integrating over $x$ w.r.t.\ $p_X$ yields \eqref{eq:m_bound}.
\end{proof}

\subsection{Proof of Proposition~\ref{prop:spectral_self_adjoint} (Spectral Concentration)}

\begin{proof}[Proof of Proposition~\ref{prop:spectral_self_adjoint}]
Assume the noiseless recursion $\Sigma_t=A^t\Sigma_0(A^\top)^t$ and the self-adjoint local form $A=U\mathrm{diag}(a_1,\ldots,a_d)U^\top$ with $a_1>a_2\ge\cdots\ge a_d\ge0$.
Since $A$ is symmetric positive semidefinite,
\begin{equation}
A^t=U\mathrm{diag}(a_1^t,\ldots,a_d^t)U^\top .
\end{equation}
With $\widetilde{\Sigma}_0=U^\top\Sigma_0U=\mathrm{diag}(\sigma_1^2,\ldots,\sigma_d^2)$, this gives
\begin{equation}
\begin{aligned}
\Sigma_t
&=U\mathrm{diag}(a_1^t,\ldots,a_d^t)
\widetilde{\Sigma}_0
\mathrm{diag}(a_1^t,\ldots,a_d^t)U^\top\\
&=U\mathrm{diag}(\sigma_1^2a_1^{2t},\ldots,\sigma_d^2a_d^{2t})U^\top .
\end{aligned}
\end{equation}
The top eigenvalue is lower bounded by the energy in the first eigen direction:
\begin{equation}
\lambda_1(\Sigma_t)\ge \sigma_1^2a_1^{2t}\ge \alpha a_1^{2t}.
\label{eq:lambda1_lb}
\end{equation}
Every direction orthogonal to $u_1$ is amplified by at most $a_2^{2t}$, so the second eigenvalue satisfies
\begin{equation}
\lambda_2(\Sigma_t)
\le \|\Sigma_0\|_2 a_2^{2t}
=\|\Sigma_0\|_2 a_1^{2t}\rho_2^{2t}.
\label{eq:lambda2_ub}
\end{equation}
Combining \eqref{eq:lambda1_lb}--\eqref{eq:lambda2_ub} gives the explicit ratio bound
\begin{equation}
\frac{\lambda_1(\Sigma_t)}{\lambda_2(\Sigma_t)}
 \ge 
\frac{\alpha}{\|\Sigma_0\|_2\rho_2^{2t}}.
\label{eq:ratio_bound}
\end{equation}
If $\lambda_2(\Sigma_t)=0$, the ratio is infinite and the same conclusion holds. In particular, when $a_1\gg a_2$ (so $\rho_2=a_2/a_1\ll1$), the ratio grows at least as fast as $\rho_2^{-2t}$, yielding spectral concentration along the leading eigen direction $u_1$.
\end{proof}

\section{Quality Gate Set: Full Definition}
\label{app:quality_feasible}

\subsection{Quality Dimensions for Code Solutions}

Let $q$ denote a task/query and let $\mathcal{D}_q$ be the induced input distribution (test-time instances) for $q$.
For a candidate program $x\in\mathcal{X}$, let $R(x,i)$ denote the execution output on input $i\sim\mathcal{D}_q$, and let $y(q,i)$ denote the task-specific semantic ground truth. 
We define the following measurable quality dimensions, each of which can be used either as an evaluation metric or as a constraint in a quality gate set.

\paragraph{(A) Semantic correctness.}\label{A}
We quantify semantic success by the probability of producing the correct semantic output under $\mathcal{D}_q$:
\begin{equation}
m_{\mathrm{corr}}(x;q)
 := 
\Pr_{i\sim\mathcal{D}_q} \bigl[R(x,i)=y(q,i)\bigr].
\end{equation}
When only a reference answer is available (instead of $y(q,i)$), the event $R(x,i)=y(q,i)$ can be replaced by ``execution matches the reference after normalization.''

\paragraph{(B) Runtime safety (no crashes/undefined behavior).}
Let $\mathsf{Crash}(x,i)$ be an indicator of runtime failure on input $i$ (e.g., uncaught exceptions, out-of-bounds access, non-termination detected by a timeout, or undefined behavior in low-level languages). We define
\begin{equation}
m_{\mathrm{safe}}(x;q)
 := 
1-\Pr_{i\sim\mathcal{D}_q} \bigl[\mathsf{Crash}(x,i)\bigr].
\end{equation}
Equivalently, one may impose a hard safety constraint:
\begin{equation}
\Pr_{i\sim\mathcal{D}_q} \bigl[\mathsf{Crash}(x,i)\bigr]\le \delta_{\mathrm{safe}},
\end{equation}
for a user-specified tolerance $\delta_{\mathrm{safe}}\in(0,1)$.

\paragraph{(C) Robustness to perturbations and distribution shift.}
Let $\mathcal{T}(i)$ denote a set of admissible perturbations of an input $i$ (e.g., format variations, paraphrases, boundary-value modifications). We define robust correctness as
\begin{equation}
\fitmath{\begin{aligned}
m_{\mathrm{rob}}(x;q)
&:=\Pr_{i\sim\mathcal{D}_q}
\Bigl[\forall i'\in\mathcal{T}(i),\ R(x,i')=y(q,i')\Bigr].
\end{aligned}}
\end{equation}
Alternatively, using a loss $\ell(\cdot,\cdot)$, we define a worst-case risk:
\begin{equation}
\fitmath{\begin{aligned}
\mathcal{R}_{\mathrm{rob}}(x;q)
&:=\mathbb{E}_{i\sim\mathcal{D}_q}
\Bigl[\sup_{i'\in\mathcal{T}(i)} \ell\bigl(R(x,i'),y(q,i')\bigr)\Bigr].
\end{aligned}}
\end{equation}

\paragraph{(D) Resource usage and efficiency.}
Let $T(x,i)$ be the runtime cost (e.g., wall-clock time) and optionally $M(x,i)$ the memory footprint or $C(x,i)$ the number of external API calls. A simple efficiency score is the negative expected runtime:
\begin{equation}
m_{\mathrm{res}}(x;q)
 := 
-\mathbb{E}_{i\sim\mathcal{D}_q} \bigl[T(x,i)\bigr].
\end{equation}
One may also enforce probabilistic latency constraints:
\begin{equation}
\Pr_{i\sim\mathcal{D}_q} \bigl[T(x,i)\le \tau\bigr]\ge 1-\delta_{\mathrm{time}},
\end{equation}
or tail-latency bounds (e.g., the $0.99$-quantile):
\begin{equation}
\mathrm{Quantile}_{0.99} \bigl(T(x,i)\bigr)\le \tau_{0.99}.
\end{equation}

\paragraph{(E) Security and compliance.}
Let $\mathsf{Bad}(x,i)$ indicate unsafe or policy-violating behavior (e.g., invoking dangerous system calls, writing to sensitive paths, unintended network access, or insecure deserialization). We define
\begin{equation}
m_{\mathrm{sec}}(x;q)
 := 
1-\Pr_{i\sim\mathcal{D}_q} \bigl[\mathsf{Bad}(x,i)\bigr],
\end{equation}
or impose a hard constraint $\Pr[\mathsf{Bad}]\le \delta_{\mathrm{sec}}$ for $\delta_{\mathrm{sec}}\in(0,1)$.

\paragraph{(F) Maintainability and testability.}
Maintainability can be captured by static code metrics such as cyclomatic complexity, code duplication, and dependency burden:
\begin{equation}
\fitmath{\begin{aligned}
m_{\mathrm{maint}}(x)
&:=-\Bigl(\lambda_1 \mathrm{Cyclomatic}(x)
\ +\lambda_2 \mathrm{Dup}(x)+\lambda_3 \mathrm{Deps}(x)\Bigr),
\end{aligned}}
\end{equation}
with weights $\lambda_1,\lambda_2,\lambda_3\ge 0$.
If a unit-test distribution $\mathcal{T}_{\mathrm{unit}}$ is available, test pass rate provides an executable proxy:
\begin{equation}
m_{\mathrm{test}}(x)
 := 
\Pr_{t\sim\mathcal{T}_{\mathrm{unit}}} \bigl[x\ \text{passes test}\ t\bigr].
\end{equation}

\paragraph{(G) Readability and style (including documentation consistency).}
Let $S_{\mathrm{read}}(x)\in[0,1]$ be a readability/style score from either rule-based or learned evaluators. We set
\begin{equation}
m_{\mathrm{read}}(x)
 := 
S_{\mathrm{read}}(x).
\end{equation}
When inline documentation $\mathrm{doc}(x)$ is present, one may additionally measure documentation--implementation consistency via an alignment score $\mathrm{Align}\bigl(x,\mathrm{doc}(x)\bigr)$, where larger values indicate better semantic agreement.

\subsection{Quality-Feasible Code Set}

Let $q$ denote a task/query, $\mathcal{X}$ the program space, and $\mathcal{D}_q$ the input distribution induced by $q$.
Let $V:\mathcal{X}\to\{0,1\}$ be a feasibility verifier, $R(x,i)$ the execution result of program $x$ on input $i$, and $y(q,i)$ the task-specific semantic ground truth.
We further define runtime event predicates $\mathsf{Crash}(x,i)$ (crash/exception/undefined behavior) and $\mathsf{Bad}(x,i)$ (unsafe or policy-violating behavior), a runtime cost $T(x,i)$ (e.g., wall-clock time), and a readability/style scorer $S_{\mathrm{read}}(x)\in[0,1]$.
For user-specified thresholds $\delta_{\mathrm{safe}},\delta_{\mathrm{sec}}\in(0,1)$ and $\tau,r_0>0$, we define the \emph{quality-feasible} solution set as
\begin{equation}
\fitmath{\begin{aligned}
\mathcal{F}_{\mathrm{qual}}(q)
 := 
\Bigl\{
x\in\mathcal{X}&: 
\underbrace{V(x)=1}_{\text{feasibility}},\\
\ &\underbrace{\Pr_{i\sim\mathcal{D}_q} \bigl[R(x,i)=y(q,i)\bigr]\ge \delta_{\mathrm{corr}}}_{\text{semantic correctness}},\\
\ &\underbrace{\Pr_{i\sim\mathcal{D}_q} \bigl[\mathsf{Crash}(x,i)\bigr]\le \delta_{\mathrm{safe}}}_{\text{runtime safety}},\\
\ &\underbrace{\mathbb{E}_{i\sim\mathcal{D}_q} \bigl[T(x,i)\bigr]\le \tau}_{\text{resource/efficiency}},\\
\ &\underbrace{\Pr_{i\sim\mathcal{D}_q} \bigl[\mathsf{Bad}(x,i)\bigr]\le \delta_{\mathrm{sec}}}_{\text{security/compliance}},\\
\ &\underbrace{S_{\mathrm{read}}(x)\ge r_0}_{\text{readability/style}}
\Bigr\},
\end{aligned}}
\label{eq:quality_feasible_set}
\end{equation}
where $\delta_{\mathrm{corr}}\in(0,1]$ specifies the minimum acceptable semantic success rate. If only a reference answer is available (rather than $y(q,i)$), the event $R(x,i)=y(q,i)$ can be replaced by ``execution matches the reference after normalization.''

\subsection{Supporting Assumptions}

\begin{assumption}[Non-trivial quality-feasible completeness / coverage]\label{assump:coverage}
Let $\mathcal{D}$ denote the target distribution over queries $q$, and for each $q$ let $\mathcal{D}_q$ be the induced input distribution.
Let $V:\mathcal{X}\to\{0,1\}$ be a feasibility verifier and define the quality gate set $\mathcal{F}_{\mathrm{qual}}(q)\subseteq\mathcal{X}$ as in the main text (incorporating feasibility, semantic correctness, runtime safety, efficiency, security/compliance, and readability/style).
Then for $\mathcal{D}$-almost every query $q$, there exists a non-empty subset $\mathcal{F}_{\mathrm{qual}}(q)$ and constants $\alpha,\beta\in(0,1]$ such that
\begin{equation}
\mu \bigl(\mathcal{F}_{\mathrm{qual}}(q)\bigr)\ge \alpha,
\qquad
P_{\theta_0} \bigl(\mathcal{F}_{\mathrm{qual}}(q)\mid q\bigr)\ge \beta,
\end{equation}
where $\mu$ is a reference measure on $\mathcal{X}$ (or on the relevant solution equivalence class for $q$), and $P_{\theta_0}(\cdot\mid q)$ is the initial generator's conditional distribution.
\end{assumption}

\begin{assumption}[Quality-feasible margin]\label{assump:margin}
Let $d:\mathcal{X}\times\mathcal{X}\to\mathbb{R}_{\ge 0}$ be a metric (or pseudometric) on $\mathcal{X}$, and define
\begin{equation}
d(x,\mathcal{F}_{\mathrm{qual}}(q)):=\inf_{y\in\mathcal{F}_{\mathrm{qual}}(q)} d(x,y).
\end{equation}
There exists $\gamma>0$ such that for $\mathcal{D}$-almost every $q$, for all $x\in\mathcal{F}_{\mathrm{qual}}(q)$ and all $x'\in\mathcal{X}$,
\begin{equation}
d(x,x')\le \gamma
 \Longrightarrow 
x'\in\mathcal{F}_{\mathrm{qual}}(q).
\end{equation}
Equivalently, the $\gamma$-ball around any $x\in\mathcal{F}_{\mathrm{qual}}(q)$ is contained in $\mathcal{F}_{\mathrm{qual}}(q)$.
\end{assumption}

\begin{assumption}[Local contraction toward the quality gate set]\label{assump:contraction}
Let $C:\mathcal{X}\to\mathcal{X}$ denote a correction (or calibration) operator applied to candidate programs.
There exist constants $\rho\in(0,1)$ and $\eta\ge 0$ such that for $\mathcal{D}$-almost every $q$ and all $x$ in the region of interest,
\begin{equation}
d\bigl(C(x),\mathcal{F}_{\mathrm{qual}}(q)\bigr)
\ \le\
\rho\, d\bigl(x,\mathcal{F}_{\mathrm{qual}}(q)\bigr)
\ +\
\eta.
\end{equation}
When $\eta=0$, $C$ is a strict contraction toward $\mathcal{F}_{\mathrm{qual}}(q)$; when $\eta>0$, the term $\eta$ captures an irreducible residual error (noise floor).
\end{assumption}







\section{AI Self-Verification Protocol}

\label{app:selfeval}

\subsection{Implemented AI Self-Gates}

\label{app:selfeval_implemented}

This appendix details the two AI self-gate methods used in our experiments: \textbf{Perplexity} and \textbf{Binary Classifier}. We use these methods as representative AI self-verification protocols because both derive acceptance scores from the code LLM itself, rather than from human labels, execution tests, or external static analyzers. Thus, the same model family that generates candidate code also evaluates which samples enter future training.

\paragraph{Basic formulation.}
Let $q$ denote the programming problem and let $x$ denote a code sample generated by the current code LLM. We construct an evaluation prompt $\Pi_{\mathrm{eval}}(q,x)$ that contains the problem statement, the candidate code, and an explicit request for assessment. The model then produces an evaluation response
\begin{equation}
r \sim p_\theta(r \mid \Pi_{\mathrm{eval}}(q,x)).
\label{eq:selfeval_response}
\end{equation}
This is a self-evaluation procedure because the same parameterized model $p_\theta$ first generates $x$ and then, under a different prompt or scoring rule, evaluates $x$.

\paragraph{Binary-classifier self-gate.}
The binary self-gate forces the model into a two-label judgment. The acceptance score is
\begin{equation}
\fitmath{\begin{aligned}
S_{\mathrm{bin}}(x\mid q)
&=\log P_\theta(\texttt{good}\mid q,x)
-\log P_\theta(\texttt{bad}\mid q,x).
\end{aligned}}
\label{eq:selfeval_binary_score}
\end{equation}
Samples are ranked by $S_{\mathrm{bin}}(x\mid q)$, and the top 25\% are retained for the next round of recursive fine-tuning. This uses only the model parameters: no human label, execution result, or external quality metric is used.

\paragraph{Perplexity self-gate.}
Code LLMs are generative models trained to predict code tokens and are commonly used for code completion or masked/fill-in-the-middle generation. We therefore use the model's own conditional likelihood as an AI self-gate. For each generated sample $x$ under problem context $q$, we compute
\begin{equation}
S_{\mathrm{ppl}}(x\mid q)=-\frac{1}{|x|}\sum_{i=1}^{|x|}\log P_\theta(x_i\mid q,x_{<i}).
\label{eq:selfeval_ppl_score}
\end{equation}
Equivalently, the perplexity is
\begin{equation}
\mathrm{PPL}_\theta(x\mid q)=\exp\big(S_{\mathrm{ppl}}(x\mid q)\big).
\label{eq:selfeval_ppl}
\end{equation}
Generated samples with lower $S_{\mathrm{ppl}}(x\mid q)$, or lower $\mathrm{PPL}_\theta(x\mid q)$, are treated as more acceptable. In our implementation, samples are ranked by perplexity and the lowest 25\% are retained for future training. This is an AI self-gate because the same code LLM family generates $x$ through completion/filling and then scores $x$ by its own token likelihood.

\subsection{Other Candidate AI Self-Verifier Protocols}

\label{app:selfeval_candidates}

The following protocols are natural extensions of AI self-verification, but are not used as main filtering methods in our experiments.

\paragraph{Chain-of-thought judge followed by forced verdict.}
A stronger AI verifier can first generate an assessment trace and then output a forced final verdict. The evaluation prompt may request the model to summarize the intended algorithm, identify possible bugs or edge-case failures, and then output a final label from a fixed set. The model generates a reasoning trace $z$ and a verdict $y$:
\begin{equation}
p_\theta(y,z\mid q,x)=p_\theta(z\mid q,x)p_\theta(y\mid q,x,z).
\label{eq:selfeval_cot_judge}
\end{equation}

\paragraph{Self-consistency over judge samples.}
Since one judge response can be noisy, the model can sample multiple evaluations:
\begin{equation}
r^{(1)},r^{(2)},\dots,r^{(M)} \sim p_\theta(\cdot\mid \Pi_{\mathrm{eval}}(q,x)).
\label{eq:selfeval_multi_samples}
\end{equation}
The final self-judgment can be obtained by majority vote,
\begin{equation}
\hat{y}_{\mathrm{vote}}(x\mid q)=\operatorname{mode}\{y^{(1)},\dots,y^{(M)}\},
\label{eq:selfeval_vote}
\end{equation}
or by averaged label probabilities:
\begin{equation}
\hat{y}_{\mathrm{avg}}(x\mid q)=\arg\max_{y\in\mathcal{Y}}\frac{1}{M}\sum_{m=1}^{M}P_\theta^{(m)}(y\mid q,x).
\label{eq:selfeval_avg}
\end{equation}

\paragraph{Constrained multi-label judge.}
For code LLMs, the evaluation prompt should be constrained. A weak prompt such as ``Is this code good?'' is underspecified. A more controlled prompt asks the model to check whether the implementation matches the specification, handles corner cases, and avoids likely bugs, then end with exactly one label from a fixed set such as \texttt{CORRECT}, \texttt{LIKELY\_CORRECT}, \texttt{UNCERTAIN}, \texttt{LIKELY\_BUGGY}, or \texttt{BUGGY}. This still remains an AI self-gate if the score is produced by the same model family that generated the code.

\paragraph{Why AI self-evaluation degrades.}
LLM evaluators exhibit significant self-preference bias~\cite{panickssery2024selfpreference}: they
assign higher scores to outputs in their own generation style, regardless of actual quality. As the model collapses and its style converges to a narrow distribution, the evaluator's
preferences converge with it, causing the gate to accept progressively worse outputs with progressively higher confidence.
This is the self-confirming regime formalized in Assumption~\ref{assump:self_confirming}
and Theorem~\ref{thm:gating_degenerate}.
For production-quality judging protocols that mitigate this bias,
see~\cite{panickssery2024selfpreference,hashemi2024llmrubric,zheng2023judging}.


\section{Mixed Training}

\label{app:training_data}

Code LLMs differ in their pretraining data mixtures, and we adapt the recursive fine-tuning format accordingly. The key distinction is whether a model was trained mainly on raw source code or on a broader mixture of code and code-related natural language.

\paragraph{SantaCoder.}
SantaCoder~\cite{benallal2023santacoder} is pretrained on Python, Java, and JavaScript from The Stack~\cite{kocetkov2022stack}. Its data consists primarily of source-code files, with natural language appearing only inside comments, docstrings, and identifiers. We therefore fine-tune SantaCoder using pure code samples in the recursive loop.

\paragraph{StarCoder and StarCoder2.}
StarCoder~\cite{li2023starcoder} and StarCoder2~\cite{lozhkov2024starcoder2} are trained on broader code corpora that include source code together with code-related natural language, such as issues, commits, pull-request text, and notebook content. For these models, recursive fine-tuning uses a mixed code--natural-language format to better match the original training distribution.

\paragraph{Qwen2.5-Coder.}
Qwen2.5-Coder~\cite{hui2024qwen25coder} follows a multi-stage code-model training pipeline, including large-scale code pretraining, code-related natural-language data, and instruction tuning. We therefore also use a mixed code--natural-language fine-tuning format for Qwen2.5-Coder rather than pure code-only retraining.


\section{Compile Error Analysis}

\label{app:exec_pass}

Figure~\ref{fig:exec_pass} and Table~\ref{tab:exec_pass_analysis} report the compile and execution-pass analysis of generated code under vanilla recursion and the compile-based exogenous gate. We evaluate 500 generated samples with a 5s timeout. Adding the exogenous compile gate substantially increases the relaxed pass rate: vanilla drops from 55.4\% at R1 to 34.6\% at R5, while compile filtering maintains a much higher relaxed pass rate of 88.2\% at R1 and 89.4\% at R10. The error distribution also changes sharply. Under vanilla recursion, the ``Other error'' category grows from 42.8\% to 64.2\%, indicating increasing runtime failure. In contrast, the compile-based exogenous gate keeps ``Other error'' low, from 9.2\% at R1 to 7.4\% at R10. This shows that even a simple exogenous compile gate significantly improves executable validity and prevents runtime-error proliferation.

\begin{figure}[!ht]
\centering

\begin{minipage}[t]{0.49\linewidth}
  \centering
  \includegraphics[height=0.25\textheight,width=\linewidth,keepaspectratio]{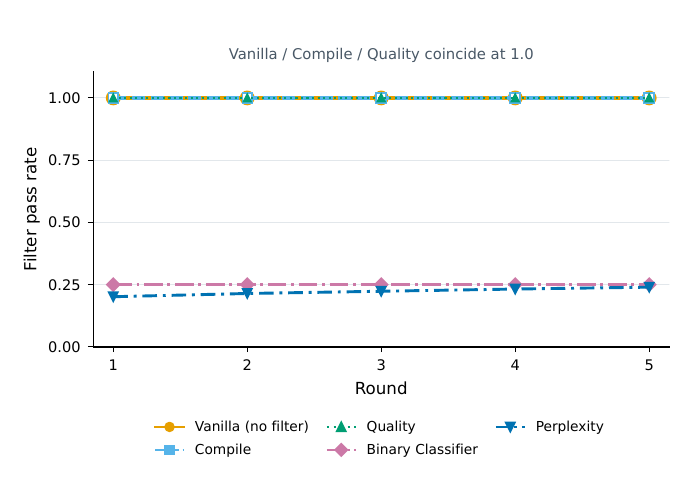}
  \caption{%
    \textbf{Filter pass rate over rounds, averaged across all four models.}
    No-gate and exogenous-gate regimes remain stable by construction or fixed rules, while AI self-gates can lose discriminative power over recursive rounds.
  }
  \label{fig:filter_pass_rate}
\end{minipage}
\hspace{0.00\linewidth}
\begin{minipage}[t]{0.49\linewidth}
  \centering
  \includegraphics[height=0.25\textheight,width=\linewidth,keepaspectratio]{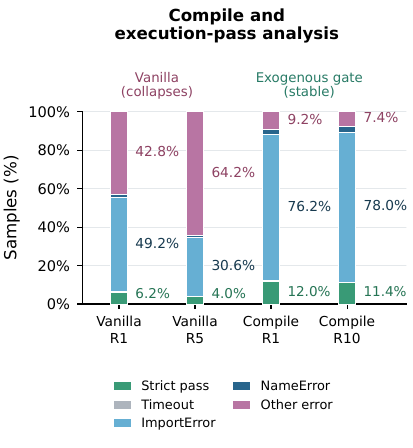}
  \caption{%
    \textbf{Compile and execution-pass analysis across recursive retraining rounds.}
    Results compare vanilla recursion with the compile-based exogenous gate on 500 generated samples with a 5s timeout.
  }
  \label{fig:exec_pass}
\end{minipage}

\end{figure}

\begin{table}[!ht]
\centering
\caption{Compile and execution-pass analysis on 500 generated samples with a 5s timeout.
The table compares ungated vanilla recursion with the compile-based exogenous gate across strict pass, relaxed pass, and error categories.}

\label{tab:exec_pass_analysis}
\resizebox{\columnwidth}{!}{
\begin{tabular}{lcccccc}
\toprule
\hline
\textbf{Dataset} & \textbf{ok (strict)} & \textbf{Timeout} & \textbf{ImportError}
  & \textbf{NameError} & \textbf{Other error} & \textbf{Relaxed pass} \\
\midrule
Vanilla R1             & 6.2\%  & 0.2\% & 49.2\% & 1.6\% & 42.8\% & 55.4\% \\
Vanilla R5 (collapsed) & 4.0\%  & 0.2\% & 30.6\% & 1.0\% & 64.2\% & 34.6\% \\
Compile R1            & 12.0\% & 0.2\% & 76.2\% & 2.4\% &  9.2\% & 88.2\% \\
Compile R10           & 11.4\% & 0.0\% & 78.0\% & 3.2\% &  7.4\% & 89.4\% \\
\hline
\bottomrule
\end{tabular}
}

\end{table}

\section{SantaCoder: All Strategies Across Rounds}

\label{subsec:santacoder}

\paragraph{Main collapse result.}
All strategies decline from the pre-retraining SantaCoder baseline
(HumanEval+ $0.171$, MBPP+ $0.294$) across 5 recursive rounds.
Figure~\ref{fig:collapse_main} and Table~\ref{tab:final_comparison} summarize the results.
Vanilla collapses most severely: HumanEval+ falls from $0.092$ at R1 to $0.073$ at R5, while
MBPP+ reaches $0.008$ at R4 and remains only $0.019$ at R5 (from baseline $0.294$; $-93.5\%$ at R5).
Binary Classifier achieves the highest R5 HumanEval+ ($0.104$) but shows erratic MBPP+
behavior, dropping to $0.019$ at R4.
Quality filtering provides the most stable MBPP+ trajectory across rounds.
\begin{figure}[!ht]
  \centering
  \includegraphics[width=\linewidth]{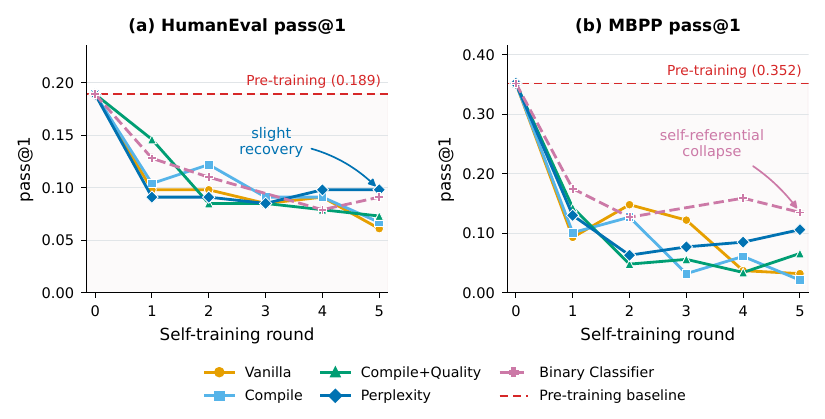}
  
  \caption{%
    \textbf{Recursive self-training results on SantaCoder.}
    HumanEval pass@1 (left) and MBPP pass@1 (right) over rounds~0--5.
    Exogenous-gate methods are Compile and Compile+Quality; AI-self-gate methods are Perplexity and Binary Classifier.
    The red dashed line denotes the SantaCoder pretraining baseline.
  }
  \label{fig:collapse_main}

\end{figure}

\paragraph{Per-round results for SantaCoder.}
Table~\ref{tab:santacoder_rounds} reports HumanEval+, MBPP+, LiveCodeBench, and filter pass rate
for SantaCoder across all five strategies and five rounds.
Figure~\ref{fig:strategy_comparison} complements the table by contrasting the early checkpoint with the extended SantaCoder checkpoint, making the short-horizon and long-horizon effects of each filter visible side by side.

\begin{table}[!ht]
\centering
\caption{SantaCoder performance across recursive rounds. HE+ = HumanEval+ pass@1, MBPP+ = MBPP+ pass@1, LCB = LiveCodeBench pass@1, and FPR = filter pass rate. Round~0 baselines are HE+=0.171, MBPP+=0.294, and LCB=0.018.}
\label{tab:santacoder_rounds}
\resizebox{0.6\columnwidth}{!}{%
\begin{tabular}{llcccc}
\toprule
\textbf{Filter} & \textbf{Round} & \textbf{HE+} & \textbf{MBPP+} & \textbf{LCB} & \textbf{FPR} \\
\midrule
\multirow{5}{*}{Vanilla}
& R1 & 0.092 & 0.172 & 0.000 & 1.000 \\
& R2 & 0.092 & 0.109 & 0.000 & 1.000 \\
& R3 & 0.073 & 0.074 & 0.000 & 1.000 \\
& R4 & 0.067 & 0.008 & 0.000 & 1.000 \\
& R5 & 0.073 & 0.019 & 0.000 & 1.000 \\
\midrule
\multirow{5}{*}{Compile}
& R1 & 0.092 & 0.103 & 0.000 & 1.000 \\
& R2 & 0.067 & 0.053 & 0.000 & 1.000 \\
& R3 & 0.055 & 0.082 & 0.000 & 1.000 \\
& R4 & 0.098 & 0.037 & 0.000 & 1.000 \\
& R5 & 0.079 & 0.034 & 0.000 & 1.000 \\
\midrule
\multirow{5}{*}{Quality}
& R1 & 0.092 & 0.183 & 0.000 & 1.000 \\
& R2 & 0.073 & 0.058 & 0.000 & 1.000 \\
& R3 & 0.067 & 0.034 & 0.000 & 1.000 \\
& R4 & 0.067 & 0.037 & 0.000 & 1.000 \\
& R5 & 0.061 & 0.069 & 0.000 & 1.000 \\
\midrule
\multirow{5}{*}{PPL}
& R1 & 0.079 & 0.143 & 0.000 & 0.167 \\
& R2 & 0.073 & 0.101 & 0.003 & 0.201 \\
& R3 & 0.073 & 0.042 & 0.000 & 0.220 \\
& R4 & 0.073 & 0.066 & 0.000 & 0.229 \\
& R5 & 0.061 & 0.032 & 0.000 & 0.235 \\
\midrule
\multirow{5}{*}{Binary}
& R1 & 0.110 & 0.061 & 0.000 & 0.250 \\
& R2 & 0.079 & 0.061 & 0.000 & 0.250 \\
& R3 & 0.073 & 0.048 & 0.000 & 0.250 \\
& R4 & 0.061 & 0.019 & 0.000 & 0.250 \\
& R5 & \textbf{0.104} & 0.058 & 0.000 & 0.250 \\
\bottomrule
\end{tabular}%
}
\end{table}

\begin{figure}[!ht]
  \centering
  \includegraphics[width=\linewidth]{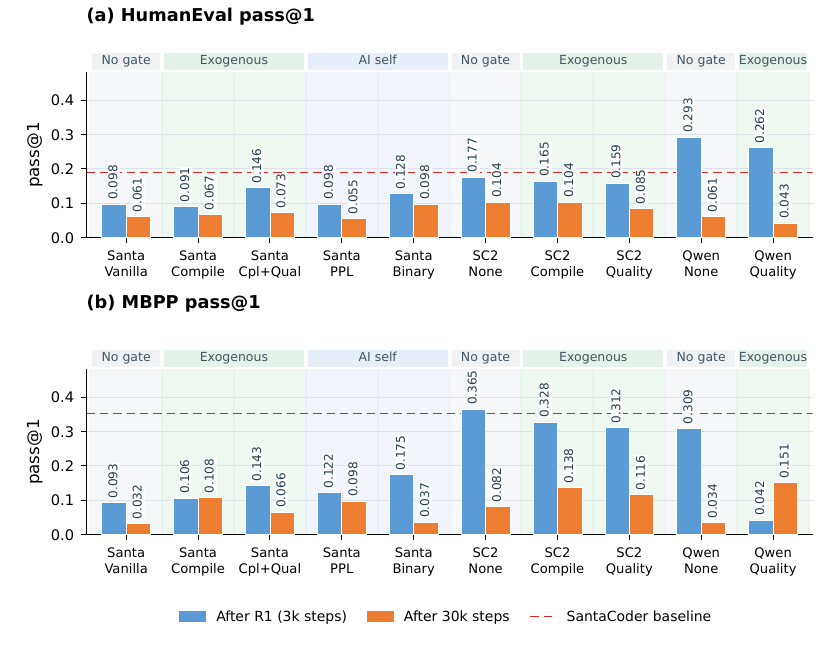}
  
  \caption{%
    \textbf{Filtering strategies across models: early vs.\ extended recursive training.}
    Blue bars show pass@1 after round~1; orange bars show the extended 30k-step checkpoint where available.
    Gate labels appear above the panels; background shading is gray for no gate, green for exogenous gates, and light blue for AI self-gates. Red dashed lines mark the SantaCoder baseline.
  }
  \label{fig:strategy_comparison}

\end{figure}
\begin{table}[!ht]
\centering
\caption{SantaCoder: all strategies at R1 and R5 (15k steps). Baseline: HE=0.189, MBPP=0.352.
\textbf{Bold} = best at R5; \underline{underline} = second best per column.}

\label{tab:final_comparison}
\resizebox{\columnwidth}{!}{
\begin{tabular}{llcccccc}
\toprule
\hline
\textbf{Method} & \textbf{Gate type} & \textbf{HE R1} & \textbf{MBPP R1} & \textbf{LCB R1}
  & \textbf{HE@R5} & \textbf{MBPP@R5} & \textbf{LCB@R5} \\
\midrule
Baseline (SantaCoder) & -- & 0.189 & 0.352 & 0.018 & -- & -- & -- \\
\midrule
Vanilla            & -- (no gate)  & 0.098 & 0.201 & 0.000 & 0.079 & 0.027 & 0.000 \\
Compile            & Exogenous gate & 0.110 & 0.122 & 0.000 & 0.092 & \underline{0.045} & 0.000 \\
Quality            & Exogenous gate & 0.116 & 0.212 & 0.000 & \underline{0.067} & \textbf{0.079} & 0.000 \\
Perplexity         & AI self-gate  & 0.104 & 0.146 & 0.000 & 0.092 & 0.048 & 0.000 \\
Binary Classifier  & AI self-gate  & \textbf{0.134} & \underline{0.077} & 0.000
  & \textbf{0.116} & 0.079 & 0.000 \\
  \hline
\bottomrule
\end{tabular}
}

\end{table}

\paragraph{Perplexity self-gating loses calibration.}
\label{subsec:perplexity_analysis}
Among all AI self-gate methods, perplexity filtering produces a more stable trajectory than
binary classification on MBPP+.
The filter pass rate for PPL slowly increases from $0.167$ at R1 to $0.235$ at R5
(Table~\ref{tab:santacoder_rounds}), indicating gradual degradation of the self-calibration
signal (Figure~\ref{fig:filter_pass_rate}).
This aligns with Theorem~\ref{thm:gating_degenerate}: once the generating model collapses,
its perplexity signal becomes miscalibrated.


\section{Per-Round Degradation Trajectories (SantaCoder, All Strategies)}

\label{app:per_round_tables}

Figure~\ref{fig:app_strategies_degrade} reports the per-round degradation trajectories of all six filtering strategies on SantaCoder. Each panel corresponds to one strategy and shows HumanEval pass@1 and MBPP pass@1 across recursive retraining rounds. Dotted horizontal lines denote the pre-retraining baseline for each benchmark. Quality-20R is evaluated for 20 rounds with sparse checkpoints to study long-horizon behavior, while the other strategies are evaluated for 5--10 rounds.
\begin{figure}[!ht]
  \centering
  \includegraphics[width=\linewidth]{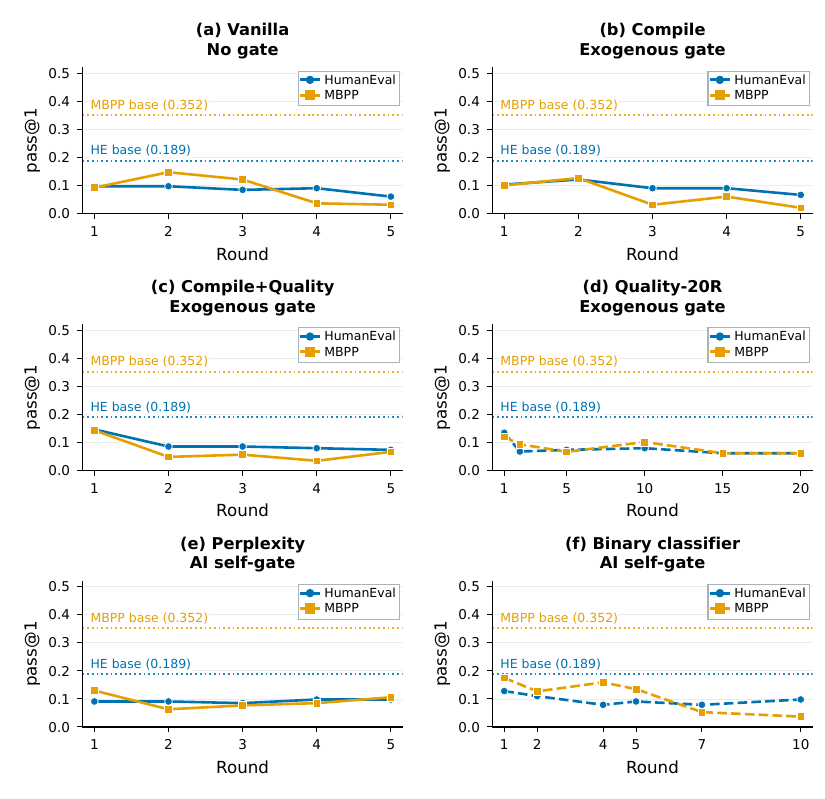}
  
  \caption{%
    \textbf{SantaCoder per-strategy degradation trajectories.}
    Each panel shows one filtering strategy.
    HumanEval pass@1 is shown with blue circles, and MBPP pass@1 is shown with orange squares.
    Dotted horizontal lines denote the pre-retraining baseline for each benchmark.
    The panels include Vanilla, Compile, Compile+Quality, Quality-20R, Perplexity, and Binary Classifier.
  }
  \label{fig:app_strategies_degrade}

\end{figure}

Figure~\ref{fig:app_strategies_degrade} shows that all strategies degrade below the pre-retraining baseline as recursive training proceeds. Vanilla collapses rapidly without any filtering, while exogenous-gate strategies generally preserve higher scores for longer. Quality-20R remains more stable over long horizons than the AI-self-gate methods. In contrast, Binary Classifier reaches a high early MBPP score but then drops sharply, matching the self-confirming gate failure described in Theorem~\ref{thm:gating_degenerate}.


\section{Cross-Model Degradation Trajectories}
\label{app:cross_model}


\paragraph{Heatmap of HumanEval+ across rounds.}
Figure~\ref{fig:heatmap} provides a compact overview of all 20 model$\times$filter combinations
across rounds 1--5.
Brighter cells indicate higher pass@1; the progressive darkening from left to right
within each model block illustrates the general collapse trend.
StarCoder2-3B with binary filtering achieves the highest round-1 score ($0.195$) among all
combinations, while Qwen2.5-Coder-1.5B with vanilla self-training collapses most severely,
dropping to $0.043$ at Round~5.
\begin{figure}[!ht]
  \centering
  \includegraphics[width=\textwidth]{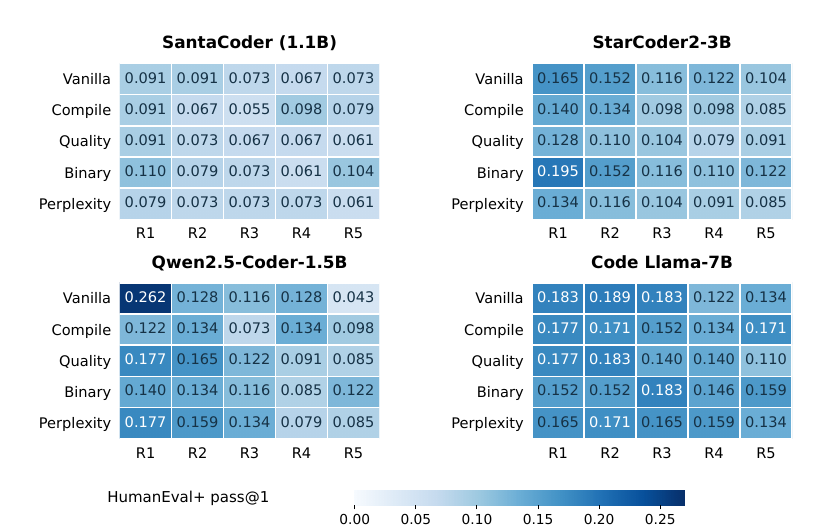}
  
  \caption{%
    \textbf{HumanEval+ pass@1 heatmap over recursive rounds.}
    Each panel shows one model; rows are filters and columns are rounds~1--5.
    Darker cells indicate higher pass@1, and annotations give the values.
    All panels use the same color scale.
  }
  \label{fig:heatmap}

\end{figure}

\subsection{MBPP+ Collapse Pattern}

\label{subsec:mbpp_collapse}

MBPP+ degrades more severely than HumanEval+ across all models and filter strategies.
Figure~\ref{fig:mbpp_collapse} shows MBPP+ trajectories for all models under all five filters.
The most dramatic collapse is observed for SantaCoder under Vanilla self-training,
where MBPP+ drops from the baseline $0.294$ to $0.008$ at Round~4 ($-97.3\%$).
Code Llama-7B is notably more robust to MBPP+ collapse, with MBPP+ under Vanilla remaining
at $0.214$ at Round~5, reflecting its stronger pretraining foundation.
\begin{figure}[!ht]
  \centering
  \includegraphics[width=\linewidth]{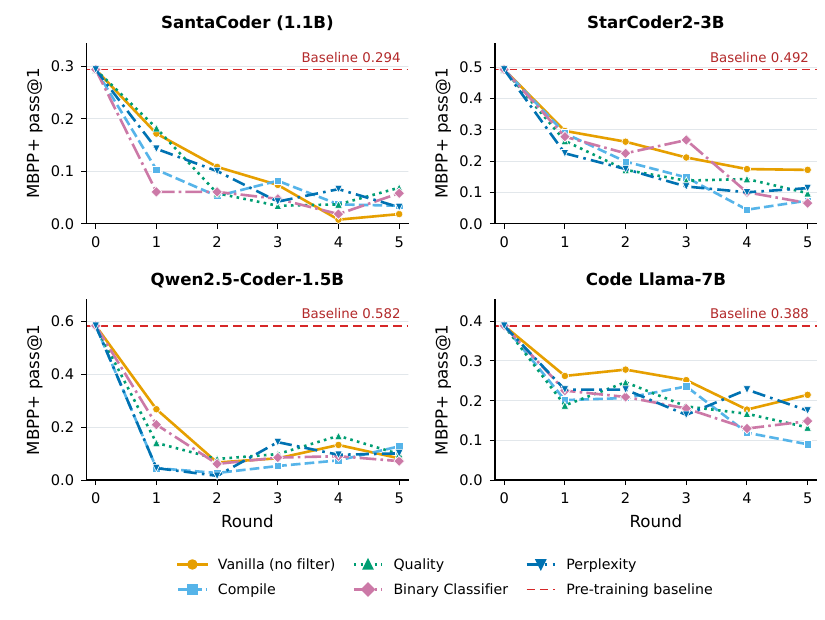}
  
  \caption{%
    \textbf{MBPP+ collapse across all models and filter strategies.}
    Each subplot shows one model; all five filter trajectories from Round~0 (baseline) to Round~5.
    Red dashed horizontal lines mark model-specific pre-training baselines.
    MBPP+ degrades more severely than HumanEval+ for all models, particularly SantaCoder and Qwen2.5-Coder.
  }
  \label{fig:mbpp_collapse}

\end{figure}

\subsection{LiveCodeBench: Near-Zero Collapse}

\label{subsec:lcb_collapse}

Figure~\ref{fig:livecodebench} shows LiveCodeBench pass@1 trajectories under the Vanilla filter
for all four models.
Despite Qwen2.5-Coder starting from the highest LCB baseline ($0.238$), it collapses to $0.000$
by Round~5.
Code Llama-7B starts at $0.045$ and reaches approximately $0.003$ at R5 under Vanilla.
StarCoder2-3B and SantaCoder are already near zero throughout.
This confirms that recursive self-training disproportionately degrades harder,
out-of-distribution reasoning tasks even when introductory benchmarks show some residual performance.

\begin{figure}[!ht]
\centering

\begin{minipage}[t]{0.56\linewidth}
  \vspace{0pt}
  \centering
  \includegraphics[width=\linewidth]{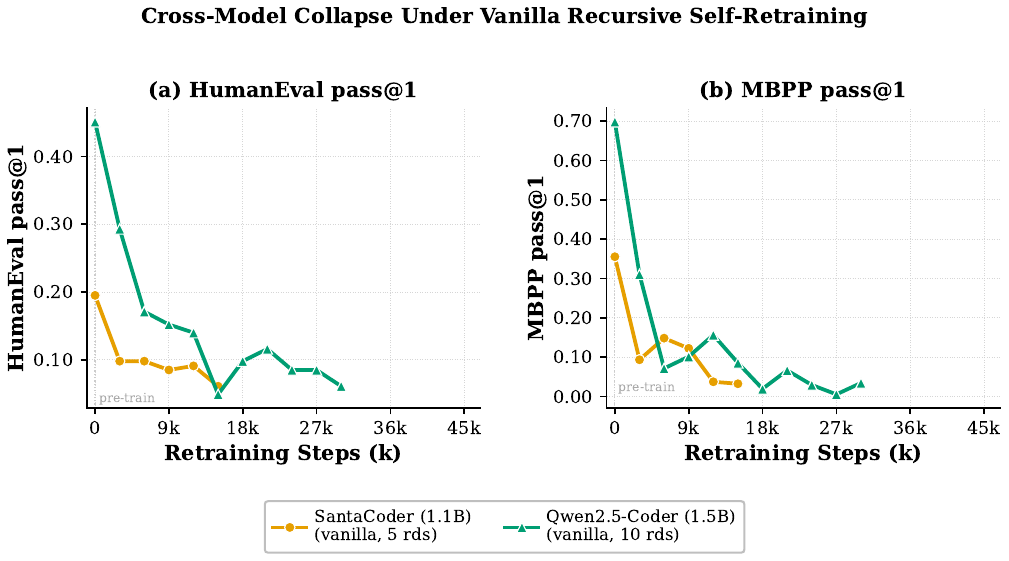}
  \captionof{figure}{%
    \textbf{Cross-model HumanEval and MBPP collapse under vanilla recursive self-retraining.}
    Per-round trajectories degrade from their pre-retraining baselines across model families.
  }
  \label{fig:cross_model_degrade}
\end{minipage}
\hspace{-0.00\linewidth}
\begin{minipage}[t]{0.39\linewidth}
  \vspace{0pt}
  \centering
  \includegraphics[width=\linewidth, height=0.25\textheight,]{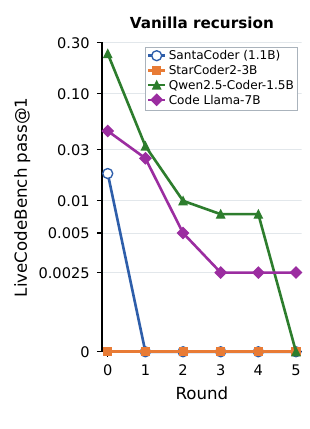}

  \captionof{figure}{%
    \textbf{LiveCodeBench collapse under Vanilla self-training.}
    All four models approach near-zero pass@1 by Round~5, providing additional evidence beyond HumanEval and MBPP.
  }
  \label{fig:livecodebench}
\end{minipage}

\end{figure}

This section reports cross-model degradation trajectories under vanilla recursion, exogenous-gate filtering, and AI-self-gate filtering. Figure~\ref{fig:strategy_multimodel} compares filtering strategies on StarCoder2-3B and Qwen2.5-Coder-1.5B. Figure~\ref{fig:app_cross_gated} further reports compile-gated runs and a genuine extended vanilla Qwen2.5-Coder trajectory.
\begin{figure}[!ht]
  \centering
  \includegraphics[width=\linewidth]{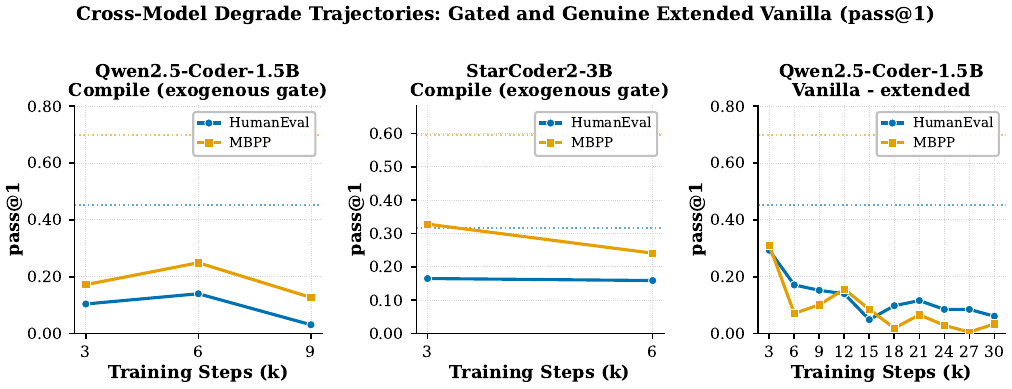}
  
  \caption{%
    \textbf{Cross-model degradation trajectories for gated and extended vanilla runs.}
    The panels report compile-gated runs for Qwen2.5-Coder-1.5B and StarCoder2-3B, plus the extended vanilla run for Qwen2.5-Coder-1.5B.
    Blue curves denote HumanEval, orange curves denote MBPP, and dotted lines denote pre-retraining baselines.
  }
  \label{fig:app_cross_gated}

\end{figure}

Figures~\ref{fig:strategy_multimodel} and \ref{fig:app_cross_gated} show that recursive self-training collapse is not specific to SantaCoder. Under vanilla recursion, the tested model families fall far below their pre-retraining baselines, with Qwen2.5-Coder dropping fastest among the shown vanilla trajectories despite starting from the highest baseline. Exogenous-gate filtering can delay degradation, especially in the early rounds, but it does not fully prevent long-horizon decline. AI-self-gate filtering by PPL is also unstable: it remains below the pretraining baseline and often underperforms compile or quality filtering on MBPP. Overall, the cross-model results support the same narrower conclusion as the main experiments: all recursive regimes degrade below baseline, exogenous gates preserve executable validity, and AI self-gates do not reliably stop recursive degradation.

Figure~\ref{fig:cross_model_degrade} compares available vanilla recursive self-training trajectories across SantaCoder, StarCoder2-3B, and Qwen2.5-Coder. All shown models degrade rapidly from their step-0 pretraining baseline on both HumanEval and MBPP. Qwen2.5-Coder starts from the highest baseline but collapses the fastest, dropping sharply within the first few retraining steps. SantaCoder also falls quickly and reaches low scores after only a few rounds. StarCoder2-3B degrades more slowly than the smaller models, but it still remains far below its initial baseline after extended retraining. These results show that larger or stronger code models may slow the collapse rate, but vanilla recursive self-training still causes consistent long-horizon degradation.

\subsection{Cross-Model Analysis}

\label{subsec:cross_model_analysis}
\begin{figure}[!ht]
  \centering
  \includegraphics[width=\linewidth]{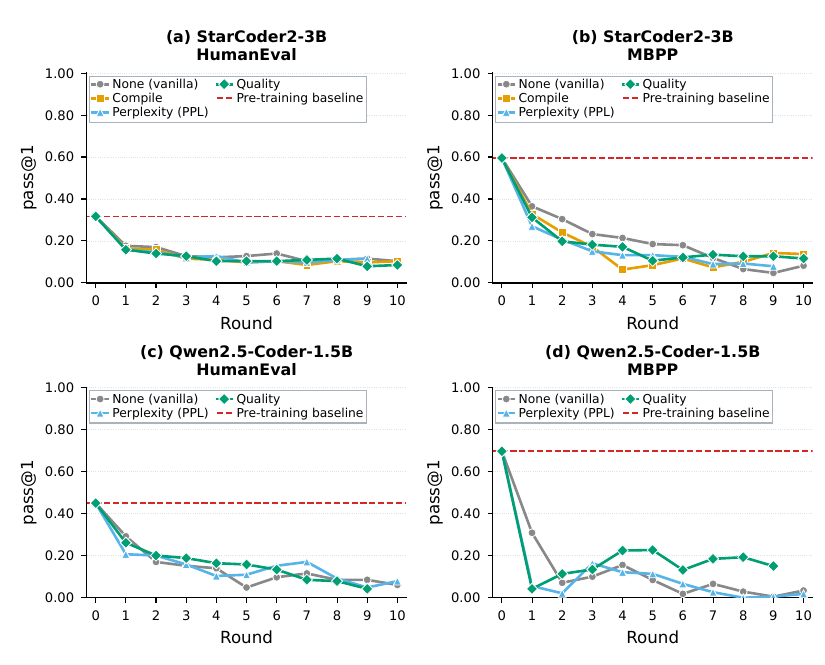}
  
  \caption{
    \textbf{Per-strategy collapse trajectories across models.}
    HumanEval and MBPP \texttt{pass@1} over recursive self-training rounds for StarCoder2-3B (top row)
    and Qwen2.5-Coder-1.5B (bottom row) under all filtering strategies.
    Dashed red line indicates each model's pre-training baseline.
    All strategies converge well below the pre-training baseline by round~5.
  }
  \label{fig:strategy_multimodel}

\end{figure}

Figure~\ref{fig:strategy_multimodel} shows that both StarCoder2-3B and Qwen2.5-Coder-1.5B exhibit recursive self-training collapse regardless
of filtering strategy, providing direct empirical corroboration that the collapse mechanism is not specific to SantaCoder.

\paragraph{Collapse magnitude.}
Under vanilla self-retraining (no filter), StarCoder2-3B declines from a HumanEval+ baseline of
$0.274$ to $0.104$ by round~5 ($-62.0\%$) and from an MBPP+ baseline of $0.492$ to $0.172$
($-65.0\%$).
The steeper decline in Qwen2.5-Coder is consistent with the stronger initial distribution
hypothesis: a model pretrained on a richer code corpus exhibits a sharper distributional
shift when contaminated with self-generated outputs.

\paragraph{Effect of filtering strategies.}
The benefit of a gate depends on which failure mode the benchmark rewards.
Compile filtering improves Qwen2.5-Coder at R5 relative to Vanilla on both HumanEval+ ($0.098$ vs.\ $0.043$) and MBPP+ ($0.127$ vs.\ $0.082$), and it is the strongest Code Llama-7B strategy on HumanEval+ ($0.171$).
For StarCoder2-3B, however, compile and quality filters lag behind Vanilla at R5 on HumanEval+ and MBPP+, showing that syntactic validity alone cannot guarantee semantic preservation.
Binary filtering achieves the highest R5 HumanEval+ for three of four models, which means self-scoring can still rank some useful samples in early and medium horizons.
Its MBPP+ behavior is weaker and less stable, matching our central story: AI self-gates may improve one visible metric while losing broader semantic coverage.
Perplexity (PPL) filtering shows a slow increase in pass rate from approximately $0.17$ to $0.24$
(averaged across models; Figure~\ref{fig:filter_pass_rate}), consistent with the endogenous degeneracy
result of Theorem~\ref{thm:gating_degenerate}.

\paragraph{Collapse attractor.}
Despite large differences in architecture size, pretraining data, and starting pass@1,
all four models show qualitatively similar long-horizon degradation patterns.
SantaCoder ($0.171$ HumanEval+ baseline) and Qwen2.5-Coder ($0.372$ baseline) both reach
approximately $0.04$--$0.10$ HumanEval+ pass@1 under Vanilla by Round~5.
StarCoder2-3B and Code Llama-7B maintain somewhat higher residual performance ($0.10$--$0.14$),
likely due to stronger capacity and pretraining data quality.
This convergence to a common degradation floor is consistent with the theoretical prediction
of a collapse attractor.

\paragraph{Key takeaway.}
No filtering strategy evaluated here---Compile, Perplexity, Quality, or Binary---prevents
collapse across all model families.
Binary filtering achieves the best R5 HumanEval+ for three models, but this comes with
erratic MBPP+ behavior.
Exogenous gates are more interpretable and preserve executable validity, yet simple compile/static filters are still too weak to guarantee semantic robustness.
LiveCodeBench approaches near-zero for all models by Round~5 regardless of strategy,
confirming that recursive self-training systematically degrades harder, out-of-distribution tasks.


\section{Implementation Details}

\label{app:impl}

\paragraph{Data generation and filtering.}
For exogenous-gate methods, we generate samples in batches of 2000 until 5000 samples pass the filter. Thus, 10 recursive rounds correspond to 30k training steps. 

\paragraph{Evaluation and hardware.}
For execution-pass analysis, we evaluate 500 generated samples with a 5s timeout. All experiments are conducted on NVIDIA RTX 6000 Ada Generation GPUs.


\section{Per-Round Full Results Tables (All Models)}

\label{app:full_per_round_tables}

This appendix provides complete per-round results for all four models across all five filtering
strategies and five rounds. All metrics are pass@1.
The full results appear in Tables~\ref{tab:app_santacoder_full}, \ref{tab:app_starcoder2_full}, \ref{tab:app_qwen25_full}, and~\ref{tab:app_codellama_full}.

\subsection{SantaCoder (1.1B)}

\begin{table}[!ht]
\centering
\caption{SantaCoder~(1.1B): full per-round results. FPR = filter pass rate.}

\label{tab:app_santacoder_full}
\resizebox{0.9\textwidth}{!}{%
\begin{tabular}{llccccccc}
\toprule
\textbf{Filter} & \textbf{Round} & \textbf{HE} & \textbf{HE+} & \textbf{MBPP} & \textbf{MBPP+} & \textbf{LCB} & \textbf{Train Loss} & \textbf{FPR} \\
\midrule
-- & R0 & 0.189 & 0.171 & 0.352 & 0.294 & 0.018 & -- & -- \\
\midrule
Vanilla & R1 & 0.0976 & 0.0915 & 0.2011 & 0.1720 & 0.000 & 0.0601 & 1.000 \\
Vanilla & R2 & 0.1037 & 0.0915 & 0.1243 & 0.1085 & 0.000 & 0.0549 & 1.000 \\
Vanilla & R3 & 0.0854 & 0.0732 & 0.0794 & 0.0741 & 0.000 & 0.0517 & 1.000 \\
Vanilla & R4 & 0.0732 & 0.0671 & 0.0106 & 0.0079 & 0.000 & 0.0478 & 1.000 \\
Vanilla & R5 & 0.0793 & 0.0732 & 0.0265 & 0.0185 & 0.000 & 0.0506 & 1.000 \\
\midrule
Compile & R1 & 0.1098 & 0.0915 & 0.1217 & 0.1032 & 0.000 & 0.0534 & 1.000 \\
Compile & R2 & 0.0915 & 0.0671 & 0.0688 & 0.0529 & 0.000 & 0.0518 & 1.000 \\
Compile & R3 & 0.0610 & 0.0549 & 0.1005 & 0.0820 & 0.000 & 0.0507 & 1.000 \\
Compile & R4 & 0.1037 & 0.0976 & 0.0503 & 0.0370 & 0.000 & 0.0571 & 1.000 \\
Compile & R5 & 0.0915 & 0.0793 & 0.0450 & 0.0344 & 0.000 & 0.0537 & 1.000 \\
\midrule
Quality & R1 & 0.1159 & 0.0915 & 0.2116 & 0.1825 & 0.000 & 0.0488 & 1.000 \\
Quality & R2 & 0.0976 & 0.0732 & 0.0635 & 0.0582 & 0.000 & 0.0564 & 1.000 \\
Quality & R3 & 0.0854 & 0.0671 & 0.0397 & 0.0344 & 0.000 & 0.0517 & 1.000 \\
Quality & R4 & 0.0976 & 0.0671 & 0.0476 & 0.0370 & 0.000 & 0.0498 & 1.000 \\
Quality & R5 & 0.0671 & 0.0610 & 0.0794 & 0.0688 & 0.000 & 0.0435 & 1.000 \\
\midrule
PPL     & R1 & 0.1037 & 0.0793 & 0.1455 & 0.1429 & 0.000 & 0.0389 & 0.167 \\
PPL     & R2 & 0.0976 & 0.0732 & 0.1164 & 0.1005 & 0.003 & 0.0383 & 0.201 \\
PPL     & R3 & 0.1098 & 0.0732 & 0.0476 & 0.0423 & 0.000 & 0.0332 & 0.220 \\
PPL     & R4 & 0.0976 & 0.0732 & 0.0847 & 0.0661 & 0.000 & 0.0382 & 0.229 \\
PPL     & R5 & 0.0915 & 0.0610 & 0.0476 & 0.0317 & 0.000 & 0.0357 & 0.235 \\
\midrule
Binary  & R1 & 0.1341 & 0.1098 & 0.0767 & 0.0608 & 0.000 & 0.0618 & 0.250 \\
Binary  & R2 & 0.1098 & 0.0793 & 0.0926 & 0.0608 & 0.000 & 0.0537 & 0.250 \\
Binary  & R3 & 0.0793 & 0.0732 & 0.0582 & 0.0476 & 0.000 & 0.0553 & 0.250 \\
Binary  & R4 & 0.0732 & 0.0610 & 0.0265 & 0.0185 & 0.000 & 0.0526 & 0.250 \\
Binary  & R5 & 0.1159 & 0.1037 & 0.0794 & 0.0582 & 0.000 & 0.0115 & 0.250 \\
\bottomrule
\end{tabular}%
}

\end{table}

\subsection{StarCoder2-3B}

\begin{table}[!ht]
\centering
\caption{StarCoder2-3B: full per-round results.}

\label{tab:app_starcoder2_full}
\resizebox{0.9\textwidth}{!}{%
\begin{tabular}{llccccccc}
\toprule
\textbf{Filter} & \textbf{Round} & \textbf{HE} & \textbf{HE+} & \textbf{MBPP} & \textbf{MBPP+} & \textbf{LCB} & \textbf{Train Loss} & \textbf{FPR} \\
\midrule
-- & R0 & 0.317 & 0.274 & 0.595 & 0.492 & 0.000 & -- & -- \\
\midrule
Vanilla & R1 & 0.1768 & 0.1646 & 0.3651 & 0.2963 & 0.000 & 0.2005 & 1.000 \\
Vanilla & R2 & 0.1707 & 0.1524 & 0.3042 & 0.2619 & 0.000 & 0.1366 & 1.000 \\
Vanilla & R3 & 0.1280 & 0.1159 & 0.2328 & 0.2116 & 0.000 & 0.1169 & 1.000 \\
Vanilla & R4 & 0.1220 & 0.1220 & 0.2143 & 0.1746 & 0.000 & 0.1026 & 1.000 \\
Vanilla & R5 & 0.1280 & 0.1037 & 0.1852 & 0.1720 & 0.000 & 0.1013 & 1.000 \\
\midrule
Compile & R1 & 0.1646 & 0.1402 & 0.3280 & 0.2884 & 0.005 & 0.0480 & 1.000 \\
Compile & R2 & 0.1585 & 0.1341 & 0.2407 & 0.1984 & 0.003 & 0.0259 & 1.000 \\
Compile & R3 & 0.1159 & 0.0976 & 0.1720 & 0.1481 & 0.003 & 0.0410 & 1.000 \\
Compile & R4 & 0.1098 & 0.0976 & 0.0635 & 0.0450 & 0.000 & 0.0393 & 1.000 \\
Compile & R5 & 0.0976 & 0.0854 & 0.0847 & 0.0741 & 0.000 & 0.0307 & 1.000 \\
\midrule
Quality & R1 & 0.1585 & 0.1280 & 0.3122 & 0.2646 & 0.008 & 0.0519 & 1.000 \\
Quality & R2 & 0.1402 & 0.1098 & 0.1984 & 0.1720 & 0.005 & 0.0409 & 1.000 \\
Quality & R3 & 0.1280 & 0.1037 & 0.1825 & 0.1376 & 0.000 & 0.0271 & 1.000 \\
Quality & R4 & 0.1037 & 0.0793 & 0.1720 & 0.1429 & 0.000 & 0.0338 & 1.000 \\
Quality & R5 & 0.1037 & 0.0915 & 0.1058 & 0.0979 & 0.003 & 0.0330 & 1.000 \\
\midrule
PPL     & R1 & 0.1585 & 0.1341 & 0.2698 & 0.2249 & 0.005 & 0.0385 & 0.244 \\
PPL     & R2 & 0.1463 & 0.1159 & 0.2063 & 0.1746 & 0.000 & 0.0342 & 0.227 \\
PPL     & R3 & 0.1220 & 0.1037 & 0.1508 & 0.1190 & 0.000 & 0.0309 & 0.232 \\
PPL     & R4 & 0.1280 & 0.0915 & 0.1323 & 0.1005 & 0.000 & 0.0287 & 0.237 \\
PPL     & R5 & 0.0976 & 0.0854 & 0.1323 & 0.1138 & 0.003 & 0.0121 & 0.242 \\
\midrule
Binary  & R1 & 0.2317 & 0.1951 & 0.3016 & 0.2778 & 0.005 & 0.0454 & 0.250 \\
Binary  & R2 & 0.1768 & 0.1524 & 0.2566 & 0.2249 & 0.000 & 0.0401 & 0.250 \\
Binary  & R3 & 0.1524 & 0.1159 & 0.2910 & 0.2672 & 0.000 & 0.0137 & 0.250 \\
Binary  & R4 & 0.1341 & 0.1098 & 0.1190 & 0.1005 & 0.000 & 0.0381 & 0.250 \\
Binary  & R5 & 0.1402 & 0.1220 & 0.0899 & 0.0661 & 0.000 & 0.0397 & 0.250 \\
\bottomrule
\end{tabular}%
}

\end{table}

\subsection{Qwen2.5-Coder-1.5B}

\begin{table}[!ht]
\centering
\caption{Qwen2.5-Coder-1.5B: full per-round results.}

\label{tab:app_qwen25_full}
\resizebox{0.9\textwidth}{!}{%
\begin{tabular}{llccccccc}
\toprule
\textbf{Filter} & \textbf{Round} & \textbf{HE} & \textbf{HE+} & \textbf{MBPP} & \textbf{MBPP+} & \textbf{LCB} & \textbf{Train Loss} & \textbf{FPR} \\
\midrule
-- & R0 & 0.451 & 0.372 & 0.698 & 0.582 & 0.238 & -- & -- \\
\midrule
Vanilla & R1 & 0.2927 & 0.2622 & 0.3095 & 0.2672 & 0.033 & 0.1091 & 1.000 \\
Vanilla & R2 & 0.1707 & 0.1280 & 0.0714 & 0.0661 & 0.010 & 0.0338 & 1.000 \\
Vanilla & R3 & 0.1524 & 0.1159 & 0.1005 & 0.0820 & 0.008 & 0.0463 & 1.000 \\
Vanilla & R4 & 0.1402 & 0.1280 & 0.1561 & 0.1323 & 0.008 & 0.0445 & 1.000 \\
Vanilla & R5 & 0.0488 & 0.0427 & 0.0847 & 0.0820 & 0.000 & 0.0403 & 1.000 \\
\midrule
Compile & R1 & 0.1463 & 0.1220 & 0.0503 & 0.0450 & 0.030 & 0.0467 & 1.000 \\
Compile & R2 & 0.1524 & 0.1341 & 0.0317 & 0.0265 & 0.008 & 0.0423 & 1.000 \\
Compile & R3 & 0.0793 & 0.0732 & 0.0608 & 0.0529 & 0.008 & 0.0441 & 1.000 \\
Compile & R4 & 0.1707 & 0.1341 & 0.0820 & 0.0741 & 0.005 & 0.0458 & 1.000 \\
Compile & R5 & 0.1037 & 0.0976 & 0.1481 & 0.1270 & 0.003 & 0.0418 & 1.000 \\
\midrule
Quality & R1 & 0.2195 & 0.1768 & 0.1561 & 0.1402 & 0.053 & 0.0501 & 1.000 \\
Quality & R2 & 0.2012 & 0.1646 & 0.0873 & 0.0794 & 0.013 & 0.0402 & 1.000 \\
Quality & R3 & 0.1463 & 0.1220 & 0.1217 & 0.0979 & 0.015 & 0.0397 & 1.000 \\
Quality & R4 & 0.1220 & 0.0915 & 0.1958 & 0.1667 & 0.008 & 0.0287 & 1.000 \\
Quality & R5 & 0.0976 & 0.0854 & 0.1243 & 0.0979 & 0.008 & 0.0302 & 1.000 \\
\midrule
PPL     & R1 & 0.2073 & 0.1768 & 0.0556 & 0.0450 & 0.020 & 0.0287 & 0.176 \\
PPL     & R2 & 0.2012 & 0.1585 & 0.0212 & 0.0159 & 0.003 & 0.0259 & 0.186 \\
PPL     & R3 & 0.1585 & 0.1341 & 0.1640 & 0.1429 & 0.000 & 0.0276 & 0.196 \\
PPL     & R4 & 0.1037 & 0.0793 & 0.1217 & 0.0952 & 0.008 & 0.0267 & 0.218 \\
PPL     & R5 & 0.1098 & 0.0854 & 0.1138 & 0.1005 & 0.005 & 0.0277 & 0.236 \\
\midrule
Binary  & R1 & 0.1646 & 0.1402 & 0.2460 & 0.2090 & 0.043 & 0.0448 & 0.250 \\
Binary  & R2 & 0.1646 & 0.1341 & 0.0714 & 0.0608 & 0.005 & 0.0459 & 0.250 \\
Binary  & R3 & 0.1585 & 0.1159 & 0.1217 & 0.0847 & 0.013 & 0.0102 & 0.250 \\
Binary  & R4 & 0.0976 & 0.0854 & 0.1032 & 0.0899 & 0.000 & 0.0439 & 0.250 \\
Binary  & R5 & 0.1463 & 0.1220 & 0.0820 & 0.0714 & 0.000 & 0.0374 & 0.250 \\
\bottomrule
\end{tabular}%
}

\end{table}

\subsection{Code Llama-7B}

\begin{table}[!ht]
\centering
\caption{Code Llama-7B: full per-round results.}

\label{tab:app_codellama_full}
\resizebox{0.9\textwidth}{!}{%
\begin{tabular}{llccccccc}
\toprule
\textbf{Filter} & \textbf{Round} & \textbf{HE} & \textbf{HE+} & \textbf{MBPP} & \textbf{MBPP+} & \textbf{LCB} & \textbf{Train Loss} & \textbf{FPR} \\
\midrule
-- & R0 & 0.328 & 0.287 & 0.441 & 0.388 & 0.045 & -- & -- \\
\midrule
Vanilla & R1 & 0.2256 & 0.1829 & 0.2989 & 0.2619 & 0.025 & 0.0383 & 1.000 \\
Vanilla & R2 & 0.2073 & 0.1890 & 0.3492 & 0.2778 & 0.005 & 0.0327 & 1.000 \\
Vanilla & R3 & 0.1890 & 0.1829 & 0.2857 & 0.2513 & 0.003 & 0.0328 & 1.000 \\
Vanilla & R4 & 0.1341 & 0.1220 & 0.2037 & 0.1772 & 0.003 & 0.0264 & 1.000 \\
Vanilla & R5 & 0.1768 & 0.1341 & 0.2354 & 0.2143 & 0.003 & 0.0276 & 1.000 \\
\midrule
Compile & R1 & 0.2195 & 0.1768 & 0.2302 & 0.2011 & 0.008 & 0.0346 & 1.000 \\
Compile & R2 & 0.2012 & 0.1707 & 0.2354 & 0.2063 & 0.000 & 0.0318 & 1.000 \\
Compile & R3 & 0.1768 & 0.1524 & 0.2725 & 0.2354 & 0.003 & 0.0268 & 1.000 \\
Compile & R4 & 0.1707 & 0.1341 & 0.1349 & 0.1190 & 0.005 & 0.0270 & 1.000 \\
Compile & R5 & 0.1890 & 0.1707 & 0.1243 & 0.0899 & 0.003 & 0.0284 & 1.000 \\
\midrule
Quality & R1 & 0.2012 & 0.1768 & 0.2143 & 0.1878 & 0.033 & 0.0359 & 1.000 \\
Quality & R2 & 0.2256 & 0.1829 & 0.3042 & 0.2460 & 0.010 & 0.0321 & 1.000 \\
Quality & R3 & 0.1707 & 0.1402 & 0.2328 & 0.1852 & 0.005 & 0.0290 & 1.000 \\
Quality & R4 & 0.1707 & 0.1402 & 0.2037 & 0.1667 & 0.010 & 0.0056 & 1.000 \\
Quality & R5 & 0.1220 & 0.1098 & 0.1429 & 0.1323 & 0.005 & 0.0278 & 1.000 \\
\midrule
PPL     & R1 & 0.2012 & 0.1646 & 0.2963 & 0.2275 & 0.010 & 0.0286 & 0.220 \\
PPL     & R2 & 0.2012 & 0.1707 & 0.2804 & 0.2275 & 0.018 & 0.0231 & 0.245 \\
PPL     & R3 & 0.2134 & 0.1646 & 0.1958 & 0.1640 & 0.000 & 0.0216 & 0.247 \\
PPL     & R4 & 0.1890 & 0.1585 & 0.2884 & 0.2275 & 0.008 & 0.0216 & 0.246 \\
PPL     & R5 & 0.1646 & 0.1341 & 0.2063 & 0.1746 & 0.000 & 0.0212 & 0.248 \\
\midrule
Binary  & R1 & 0.2012 & 0.1524 & 0.2778 & 0.2249 & 0.030 & 0.0240 & 0.250 \\
Binary  & R2 & 0.1890 & 0.1524 & 0.2646 & 0.2090 & 0.005 & 0.0325 & 0.250 \\
Binary  & R3 & 0.2256 & 0.1829 & 0.1958 & 0.1799 & 0.008 & 0.0190 & 0.250 \\
Binary  & R4 & 0.1890 & 0.1463 & 0.1481 & 0.1296 & 0.010 & 0.0276 & 0.250 \\
Binary  & R5 & 0.1829 & 0.1585 & 0.1667 & 0.1481 & 0.003 & 0.0184 & 0.250 \\
\bottomrule
\end{tabular}%
}

\end{table}

%
%


\end{document}